\documentclass[12pt,onecolumn,oneside,a4paper,peerreview]{IEEEtran}
\usepackage{cite}
\usepackage{float}
\usepackage{graphicx}
\usepackage{amsmath}
\usepackage{times}
\usepackage{graphicx}
\usepackage{latexsym}
\usepackage{graphicx}
\usepackage{bm}
\usepackage{amssymb}
\usepackage[center]{caption2}
\usepackage{stfloats}
\usepackage{cases}
\usepackage{array}
\usepackage{setspace}
\usepackage{fancyhdr}
\usepackage{color}
\usepackage{subfigure}
\usepackage{citesort}
\usepackage{multirow}
\usepackage{amsthm}
\usepackage{cases}
\usepackage{algpseudocode}
\usepackage{algorithmicx,algorithm}
\usepackage[long]{optidef}
\usepackage{epstopdf}
\usepackage{verbatim}
\usepackage{amsmath}
\allowdisplaybreaks[4]
\def\bm#1{\mbox{\boldmath $#1$}}
\usepackage{amsmath}

\theoremstyle{definition}
\newtheorem{theorem}{Theorem}

\newtheorem{corollary}{Corollary}

\newtheorem{remark}{Remark}

\renewcommand{\baselinestretch}{0.97}

\begin{document}
\title{IRS-Aided Non-Orthogonal ISAC Systems: Performance Analysis and Beamforming Design}
\author{\IEEEauthorblockN{Zhouyuan Yu, Xiaoling Hu, {\em Member, IEEE},   Chenxi Liu, {\em Senior Member, IEEE}, \\ Mugen Peng, {\em Fellow, IEEE,}  and Caijun Zhong, {\em Senior Member, IEEE}}
\thanks{Z. Yu, X. Hu, C. Liu, and M. Peng  are with the State Key Laboratory of Networking and Switching Technology, Beijing University of Posts and
Telecommunications, Beijing 100876, China (e-mail: \{yzy9912, xiaolinghu, chenxi.liu, pmg\}@bupt.edu.cn).}
\thanks{C. Zhong is with the College of Information Science and Electronic Engineering, Zhejiang University, Hangzhou 310027, China (email: caijunzhong@zju.edu.cn).}
}

\maketitle

\begin{abstract}
Intelligent reflecting surface (IRS) has shown its effectiveness in facilitating orthogonal time-division integrated sensing and communications (TD-ISAC), in which the sensing task and the communication task occupy orthogonal time-frequency resources, while the role of IRS in the more interesting scenarios of non-orthogonal ISAC (NO-ISAC) systems has so far remained unclear.
In this paper, we consider an IRS-aided NO-ISAC system, where a distributed IRS is deployed to assist concurrent communication and location sensing for a blind-zone user, occupying non-orthogonal/overlapped time-frequency resources. We first propose a modified Cramer-Rao lower bound (CRLB) to characterize the performances of both communication and location sensing in a unified manner. We further derive the closed-form expressions of the modified CRLB in our considered NO-ISAC system, enabling us to identify the fundamental trade-off between the communication and location sensing performances. In addition, by exploiting the modified CRLB, we propose a joint active and passive beamforming design algorithm that achieves a good communication and location sensing trade-off. Through numerical results, we demonstrate the superiority of the IRS-aided NO-ISAC systems over the IRS-aided TD-ISAC systems, in terms of both communication and localization performances. Besides, it is shown that the IRS-aided NO-ISAC system with random communication signals can achieve comparable localization performance to the IRS-aided localization system with dedicated positioning reference signals. Moreover, we investigate the trade-off between communication performance and localization performance and show how the performance of the NO-ISAC system can be significantly boosted by increasing the number of the IRS elements.
\end{abstract}
% \newpage

\section{Introduction}
\par Intelligent reflecting surface (IRS) has emerged as a revolutionary technology for the sixth-generation (6G) mobile communications in terms of extending coverage and enhancing spectrum efficiency \cite{9424177,9133184}. Specifically, IRS is a digitally-controlled metasurface constructed of a large array of reflecting elements that can proactively adjust the phase shifts and amplitudes of the incident signals to reshape the wireless propagation environment \cite{wu2019intelligent,gong2020toward}. Compared to traditional active relays, the IRS passively reflects the incident signal without any specialized radio-frequency (RF) processing, thereby substantially reducing hardware cost and energy consumption \cite{you2020channel,pan2021reconfigurable}. Moreover, IRS can be deployed in the environment with dense obstacles for  establishing virtual line-of-sight (VLoS) links to bypass signal blockage between transceivers, thus achieving ubiquitous coverage efficiently and cost-effectively \cite{9722893,9743440}.

\par Due to the aforementioned attractive characteristics, IRS has been extensively investigated in both communication and localization fields. As for the IRS-aided communication, through joint active and passive beamforming, the IRS can significantly enhance the signal-to-noise ratio (SNR) \cite{9066923,9362274,8811733}, spectral/energy efficiency \cite{huang2019reconfigurable,zhou2020spectral,yu2019miso,fang2020energy}, and data rate \cite{9039554,9110912,9198125,hu_location}, especially for users in the non-line-of-sight (NLoS) region. For example, the authors in \cite{8811733} revealed that in an IRS-aided communication system, the received SNR scales with the square of the number of IRS reflecting elements. The work \cite{huang2019reconfigurable} considered an energy efficiency maximization problem via joint active and passive beamforming, and demonstrated that the IRS-aided MISO communication system can achieve up to $3$ times higher energy efficiency than the conventional amplify-and-forward (AF) relay system with half-duplex operation. In \cite{9039554}, the authors formulated a data rate maximization problem. To reduce the beamforming design complexity, an IRS element grouping strategy was adopted. Then, transmit power allocation and IRS reflection coefficients were joint optimized. As for the IRS-aided localization, the location-related information, such as received signal strength (RSS), angle of arrival (AoA), angle of departure (AoD) as well as time of arrival (ToA), can be extracted from the IRS VLoS link for user localization \cite{8264743,9456027,9500281,9361184,9732186,9725255,9724202,9508883,wang2021joint,9513781,9593241,9593200}.  In \cite{8264743}, the authors first investigated the potential of using IRS for positioning purposes and analyzed the Cramer-Rao lower bound (CRLB) for localization with IRS. It was demonstrated that the localization performance improves quadratically with IRS’s surface area. With the assistance of IRSs, the authors in \cite{9456027} proposed a RSS-based multi-user localization scheme, which achieves at least 3 times lower localization error than the traditional RSS-based localization scheme without the IRS. Later on, the authors in \cite{wang2021joint} used the IRS for the localization of a blind-zone user, and proposed an angle-based positioning algorithm, which achieves centimeter-level positioning accuracy. 

\par Nevertheless, in all the aforementioned works, IRS-aided communication and localization systems are designed, independently. Recently, an IRS-aided joint communication and localization system was proposed in \cite{wang2021joint2} and \cite{he2021beyond}, where communication and localization functions are realized in a time-division integrated sensing and communication (TD-ISAC) manner, and the time allocation ratio was optimized by taking into account both communication performance and localization performance. However, in the IRS-aided TD-ISAC system, communication and localization tasks occupy orthogonal/non-overlapped time-frequency resources. For a more efficient use of the time-frequency resources, it is desirable to realize communication and localization functions in a non-orthogonal/overlapped ISAC manner, i.e., on the same time and spectrum resources. Hence, the works \cite{hu_ISAC_TWC,hu_bpa} proposed a semi-passive IRS aided non-orthogonal ISAC (NO-ISAC) system, where uplink communication and localization for a single user are realized on  non-orthogonal/overlapped time-frequency resources at the cost of equipping partial sensing elements on the IRS, and designed the sensing-assisted beamforming scheme for improving communication performance. However, few works have investigated the joint communication and location sensing performance of the IRS-aided NO-ISAC system. For the IRS-aided NO-ISAC system, communication performance and localization performance are tightly coupled to each other, but have different performance evaluation metrics, making it difficult to make a balance between them. Hence, it is required to devise a unified metric for characterizing the joint communication and localization performance, as well as to deeply investigate the trade-off between communication performance and localization performance for better guiding the design of the IRS-aided NO-ISAC system.

\par Motivated by the above issues, in this paper, we propose an IRS-aided NO-ISAC system, where a distributed passive IRS is deployed to enable concurrent location sensing and downlink communication on non-orthogonal/overlapped time-frequency resources for a blind-zone user.  A novel metric is devised for characterizing the joint communication and location sensing performance of the IRS-aided NO-ISAC system. Based on this new performance metric, we design a joint active and passive beamforming algorithm and investigate the trade-off between communication performance and localization performance. Our main contributions are summarized as follows.
\begin{itemize}
	\item We establish an IRS-aided NO-ISAC system, where data transmission and user positioning for the blind-zone user are simultaneously realized  on non-orthogonal time-frequency resources, with the assistance of a distributed passive IRS.
	\item We novelly develop a  modified CRLB metric for characterizing the joint communication and localization performance of the IRS-aided NO-ISAC system, and derive its closed-form expression, which allows us to deeply investigate the trade-off between communication and localization performances.
	\item By invoking the modified CRLB metric, we propose a cross entropy (CE)-based beamforming algorithm to balance communication performance and localization performance.
	\item Simulation results demonstrate the effectiveness of the proposed beamforming algorithm in making trade-off between communication performance and localization performance, as well as the superiority of the proposed IRS-aided NO-ISAC system to the IRS-aided TD-ISAC system in terms of both communication performance and location sensing performance. Also, it is shown that the proposed IRS-aided NO-ISAC system can achieve comparable localization performance to the IRS-aided localization system with dedicated positioning reference signals. In addition, the existence of the trade-off region and the communication/localization saturation region is revealed.
\end{itemize}

\par The remainder of this paper is organized as follows. Section~\ref{section2} introduces the system model of the IRS-aided NO-ISAC system. Section~\ref{section3} presents the performance characterization of the IRS-aided NO-ISAC system, while Section~\ref{section4} presents a CRLB-based beamforming algorithm. Numerical results are  provided in Section~\ref{section5}. Finally, Section~\ref{section6} concludes this paper.

\emph{{Notations:}} Vectors and matrices are denoted by boldface lower case and boldface upper case, respectively. The superscripts $\left( \cdot \right) ^{\mathrm{T}}$ and $\left( \cdot \right) ^{\mathrm{H}}$ denote the operations of transpose and Hermitian transpose, respectively. The Euclidean norm, absolute value and floor function are respectively denoted by $\left\| \cdot \right\|$, $\left| \cdot \right|$, and $\lfloor \cdot \rfloor $. $\mathbb{E} \left\{ \cdot \right\} $ denotes the statistical expectation. Moreover, $\mathcal{C} \mathcal{N} \left( 0,\sigma ^2 \right) $ denotes the circularly symmetric complex Gaussian (CSCG) distribution with zero mean  and variance $\sigma ^2$. For matrices,
$\left[ \cdot \right] _{ij}$ denotes the $(i,j)$-th element of a matrix, $\text{tr}\left( \cdot \right) $ represents the matrix trace, $\text{diag}\left( \cdot \right) $ denotes a square diagonal matrix with the elements in $\left( \cdot \right) $ on its main diagonal, $\mathbf{0}_{N\times M}$ and $\mathbf{1}_{N\times M}$ denote the $N\!\times\! M$ all-zero matrix and all-one matrix, respectively. For vectors, $\left[ \cdot \right] _i$ denotes the $i$-th entry of a vector. Besides, $j$ in $e^{j\theta}$ denotes the imaginary unit, and we use $\mathcal{R} \left( \cdot \right) $ to denote the real part of the argument.

\section{System Model}\label{section2}

\begin{figure}[htb]
  \centering
  \includegraphics[width=5in]{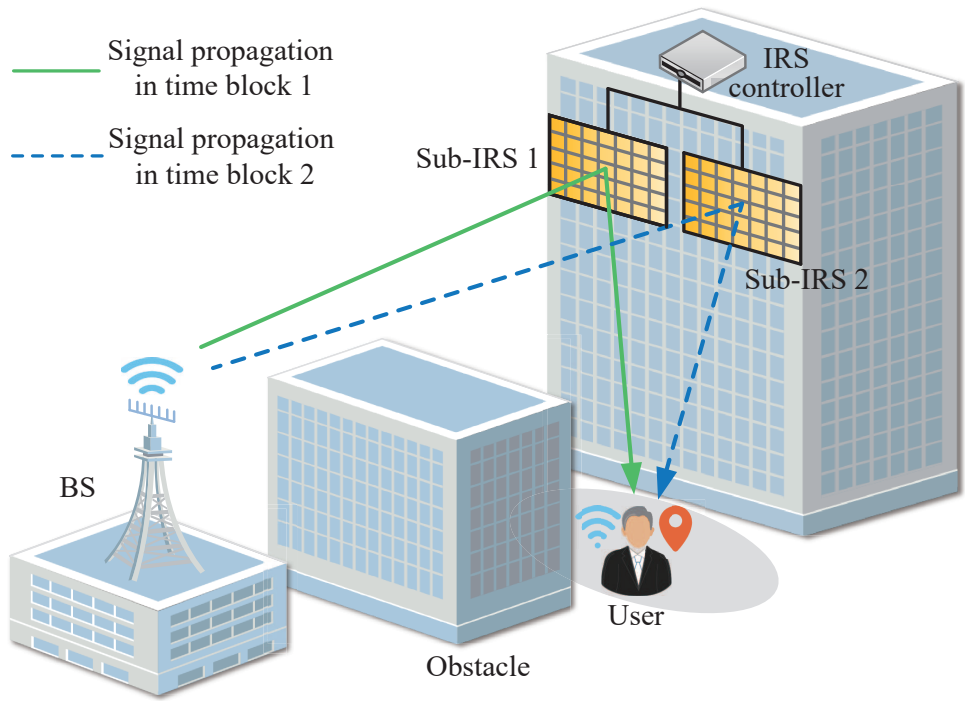}
  \caption{An IRS-aided NO-ISAC system.}
  \label{ISAC}
\end{figure}

\par As illustrated in Fig.~\ref{ISAC}, we consider an IRS-aided NO-ISAC system, where a distributed passive IRS is deployed to assist concurrent data transmission and localization for a blind-zone user. We consider that the line-of-sight (LoS) path between the BS and the user is obstructed, and the IRS is deployed to establish a strong virtual VLoS reflection path between them. The distributed passive IRS is composed of $2$ sub-IRSs, each of which has an $L=L_y\times L_z$ uniform rectangular array (URA) lying on the $y$-$o$-$z$ plane. The BS has an $N_{t}$-element uniform linear array (ULA) along the $y$ axis, while the user has an $M=M_y\times M_z$ URA lying on the $y$-$o$-$z$ plane.
\subsection{NO-ISAC Transmission Protocol}
\begin{figure}[htb]
  \centering
  \includegraphics[width=5.5in]{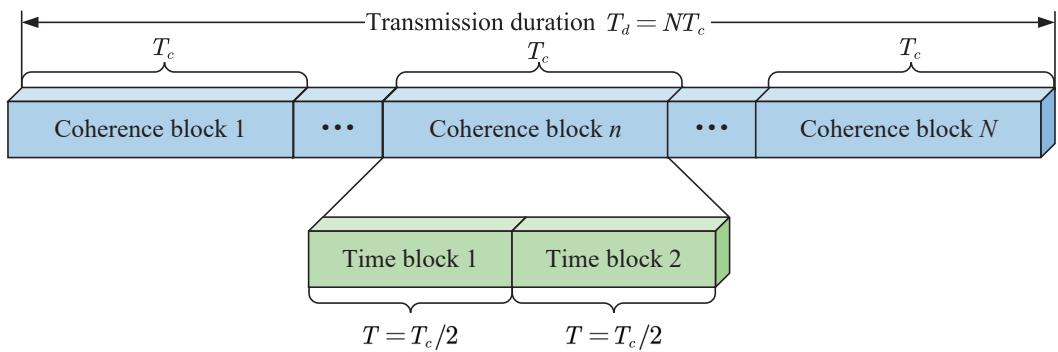}
  \caption{Illustration of the NO-ISAC transmission protocol.}
  \label{protocol}
\end{figure}
\par We consider a NO-ISAC transmission protocol as illustrated in Fig.~\ref{protocol}, where the whole transmission duration consists of $N$ consecutive coherence blocks, during each of which, both BS-IRS and IRS-user channels remain unchanged. One coherence block composed of $T_c$ time slots (symbol durations) is further divided equally into two time blocks each with $T=T_c/2$ time slots. During the $i$-th time block, the $i$-th sub-IRS assists downlink data transmission and localization for the blind-zone user by reflecting, with another sub-IRS switched off. Based on the received signals during the $i$-th time block, the user conducts demodulation for communication and AoA estimation for user positioning, simultaneously. Specifically, the $T$ information-carrying symbols  and the AoA pair from the $i$-th sub-IRS to the user are jointly estimated. For each coherence block, combing the AoA pairs estimated during its two time blocks yields the user's location.
\subsection{Signal Model}
\par Without loss of generality, we focus on the $i$-th time block of the $n$-th coherence block, during which the BS sends $x_{i}^{n}\left( t \right) $, satisfying $\mathbb{E} \left\{ |x_{i}^{n}\left( t \right) |^2 \right\} =1$ and following the complex Gaussian distribution $\mathcal{C} \mathcal{N} \left( \mu _x,\sigma _{x}^{2} \right) $, to the user at time slot $t\in \mathcal{T} \triangleq \left\{ 1,\cdots ,T \right\} $. For ease of notation, we drop the subscript $i$ and the superscript $n$ hereafter, and the received signal at the user can be expressed as
\begin{align}\label{S1}
\mathbf{y}\left( t \right) =\mathbf{H}_{\mathrm{I}2\mathrm{U}}\mathbf{\Theta H}_{\mathrm{B}2\mathrm{I}}\mathbf{w}x\left( t \right) +\mathbf{z}\left( t \right) ,t\in \mathcal{T},
\end{align}
where $\mathbf{H}_{\mathrm{I}2\mathrm{U}}\in \mathbb{C} ^{M\times L}$ and $\mathbf{H}_{\mathrm{B}2\mathrm{I}}\in \mathbb{C} ^{L\times N_t}$ denote the channels from the IRS to the user and from the BS to the IRS, respectively. The phase shift matrix of the IRS is defined as $\mathbf{\Theta }=\mathrm{diag}\left( \boldsymbol{\xi } \right)$, with the phase shift beam being $\boldsymbol{\xi }=\left[ e^{j\vartheta _1},\cdots e^{j\vartheta _l},\cdots e^{j\vartheta _L} \right] ^{\mathrm{T}}$. For ease of practical implementation, the phase shifts of the IRS take values from a finite set $\mathcal{F} =\left\{ 0,\frac{2\pi}{2^b},\cdots ,\frac{2\pi}{2^b}\left( 2^b-1 \right) \right\} $, where $b$ is the bit-quantization number. $\mathbf{w}\in \mathbb{C} ^{N_t\times 1}$ is the BS beamforming vector subject to the transmit power constraint $\left\| \mathbf{w} \right\| ^2\leqslant P_t$. In addition, $\mathbf{z}\in \mathbb{C} ^{M\times 1}$ denotes the additive white Gaussian noise (AWGN), whose elements follow the complex Gaussian distribution $\mathcal{C} \mathcal{N} \left( 0,\sigma _{z}^{2} \right)$.
\par By stacking all the  ${\bf y}(t)$'s, $t \in \mathcal{T}$ into an $MT \times 1$ vector, we obtain the signals received during one time block as
\begin{align}
\mathbf{y}=\left[ \begin{array}{c}
	\mathbf{y}\left( 1 \right)\\
	\vdots\\
	\mathbf{y}\left( T \right)\\
\end{array} \right] =\left[ \begin{array}{c}
	\mathbf{h}_x\left( 1 \right)\\
	\vdots\\
	\mathbf{h}_x\left( T \right)\\
\end{array} \right] +\left[ \begin{array}{c}
	\mathbf{z}\left( 1 \right)\\
	\vdots\\
	\mathbf{z}\left( T \right)\\
\end{array} \right] ,
\end{align}
where
\begin{align}
\mathbf{h}_x\left( t \right)=\mathbf{H}_{\mathrm{I}2\mathrm{U}}\mathbf{\Theta H}_{\mathrm{B}2\mathrm{I}}\mathbf{w}x\left( t \right).\label{mu} 
\end{align}
\subsection{Channel Model}
\par In general, the IRS operating in the mmWave band is depoloyed with LoS paths to both the BS and the user. Hence, the channel from the IRS to the user is modelled as \cite{yuanintell}
\begin{align}\label{S2}
\mathbf{H}_{\mathrm{I}2\mathrm{U}}=\alpha _{\mathrm{I}2\mathrm{U}}\mathbf{b}_{\mathrm{U}}\left( \gamma _{\mathrm{I}2\mathrm{U}}^{\mathrm{A}},\varphi _{\mathrm{I}2\mathrm{U}}^{\mathrm{A}} \right) \mathbf{b}_{\mathrm{I}}^{\mathrm{H}}\left( \gamma _{\mathrm{I}2\mathrm{U}}^{\mathrm{D}},\varphi _{\mathrm{I}2\mathrm{U}}^{\mathrm{D}} \right), 
\end{align}
where $\alpha _{\mathrm{I}2\mathrm{U}}$ denotes the complex channel gain, $\gamma _{\mathrm{I}2\mathrm{U}}^{\mathrm{A}}$/$\varphi _{\mathrm{I}2\mathrm{U}}^{\mathrm{A}}$ denotes the elevation/azimuth AoA and $\gamma _{\mathrm{I}2\mathrm{U}}^{\mathrm{D}}$/$\varphi _{\mathrm{I}2\mathrm{U}}^{\mathrm{D}}$ denotes the elevation/azimuth AoD from the IRS to the user. In addition, $\mathbf{b}_{\mathrm{U}}\left( \gamma _{\mathrm{I}2\mathrm{U}}^{\mathrm{A}},\varphi _{\mathrm{I}2\mathrm{U}}^{\mathrm{A}} \right) 
$ and $\mathbf{b}_{\mathrm{I}}\left( \gamma _{\mathrm{I}2\mathrm{U}}^{\mathrm{D}},\varphi _{\mathrm{I}2\mathrm{U}}^{\mathrm{D}} \right) $ are the array response vectors for the user and the IRS, respectively, with their elements given by
\begin{align}
&\left[ \mathbf{b}_{\mathrm{U}}\left( \gamma _{\mathrm{I}2\mathrm{U}}^{\mathrm{A}},\varphi _{\mathrm{I}2\mathrm{U}}^{\mathrm{A}} \right) \right] _m=e^{j\left[ m_y\left( m \right) 2\pi \frac{d_{\mathrm{user}}}{\lambda}\cos \left( \gamma _{\mathrm{I}2\mathrm{U}}^{\mathrm{A}} \right) \sin \left( \varphi _{\mathrm{I}2\mathrm{U}}^{\mathrm{A}} \right) +m_z\left( m \right) 2\pi \frac{d_{\mathrm{user}}}{\lambda}\sin \left( \gamma _{\mathrm{I}2\mathrm{U}}^{\mathrm{A}} \right) \right]},
\\
&\left[ \mathbf{b}_{\mathrm{I}}\left( \gamma _{\mathrm{I}2\mathrm{U}}^{\mathrm{D}},\varphi _{\mathrm{I}2\mathrm{U}}^{\mathrm{D}} \right) \right] _l=e^{j\left[ l_y\left( l \right) 2\pi \frac{d_{\mathrm{IRS}}}{\lambda}\cos \left( \gamma _{\mathrm{I}2\mathrm{U}}^{\mathrm{D}} \right) \sin \left( \varphi _{\mathrm{I}2\mathrm{U}}^{\mathrm{D}} \right) +l_z\left( l \right) 2\pi \frac{d_{\mathrm{IRS}}}{\lambda}\sin \left( \gamma _{\mathrm{I}2\mathrm{U}}^{\mathrm{D}} \right) \right]},
\end{align}
where
\begin{align}
&m_y\left( m \right) \triangleq \lfloor \left( m-1 \right) /M_z \rfloor ,m=1,\cdots ,M,
\\
&m_z\left( m \right) \triangleq m-m_y\left( m \right) M_z-1,m=1,\cdots ,M,
\\
&l_y\left( l \right) \triangleq \lfloor \left( l-1 \right) /L_z \rfloor ,l=1,\cdots ,L,
\\
&l_z\left( l \right) \triangleq l-l_y\left( l \right) L_z-1,l=1,\cdots ,L,
\end{align}
$\lambda $ denotes the carrier wavelength, $d_{\mathrm{IRS}}$ and $d_{\mathrm{user}}$ represent the distances between two adjacent reflecting elements of the IRS and two adjacent antennas of the user, respectively. Furthermore, we consider $d_\text{IRS}=d_\text{user}=\frac{\lambda}{2}$.
\par Similarly, the channel from the BS to the IRS can be expressed as
\begin{align}
\mathbf{H}_{\mathrm{B}2\mathrm{I}}=\alpha _{\mathrm{B}2\mathrm{I}}\mathbf{b}_{\mathrm{I}}\left( \gamma _{\mathrm{B}2\mathrm{I}}^{\mathrm{A}},\varphi _{\mathrm{B}2\mathrm{I}}^{\mathrm{A}} \right) \mathbf{a}_{\mathrm{B}}^{\mathrm{H}}\left( \gamma _{\mathrm{B}2\mathrm{I}}^{\mathrm{D}} \right) ,
\end{align}
where $\alpha _{\mathrm{B}2\mathrm{I}}$ denotes the complex channel gain, $\gamma _{\mathrm{B}2\mathrm{I}}^{\mathrm{A}}$/$\varphi _{\mathrm{B}2\mathrm{I}}^{\mathrm{A}}$ represents the elevation/azimuth AoA and $\gamma _{\mathrm{B}2\mathrm{I}}^{\mathrm{D}}$ represents the elevation AoD from the BS to the IRS. In addition, $\mathbf{a}_{\mathrm{B}}\left( \gamma _{\mathrm{B}2\mathrm{I}}^{\mathrm{D}} \right) $ is the array response vector for the BS.

\section{Performance Characterization of the IRS-Aided NO-ISAC System}\label{section3}
\par For the considered IRS-aided NO-ISAC system, its communication  and localization performances depend on the transmission of the information-carrying symbol $x(t)$ and the estimation of the AoA pair $\left( \gamma _{\mathrm{I}2\mathrm{U}}^{\mathrm{A}},\varphi _{\mathrm{I}2\mathrm{U}}^{\mathrm{A}} \right)$, respectively. In general, the communication performance is measured by the channel capacity, while the localization performance is measured by the CRLB. To unify the performance metric of  communication and localization, we propose the communication CRLB metric, which is defined as the CRLB for the unbiased estimator of $x(t)$.  As such, the ISAC process of simultaneous data transmission and AoA estimation is equivalent to the joint estimation of $\gamma _{\mathrm{I}2\mathrm{U}}^{\mathrm{A}}$, $\varphi _{\mathrm{I}2\mathrm{U}}^{\mathrm{A}}$, and $x\left( t \right), t=1,\cdots,T$. Motivated by this, we proposed  a modified CRLB (in dB) metric for characterizing the joint communication and localization performance of the NO-ISAC system 
\begin{align}\label{FEMSE}
\mathrm{CRLB}_\text{ISAC}\triangleq \zeta \lg \left( \mathrm{CRLB}_{\mathrm{x}} \right) +\left( 1-\zeta \right) \lg \left( \mathrm{CRLB}_{\mathrm{angle}} \right) ,
\end{align}
where
\begin{align}
&\mathrm{CRLB}_{\mathrm{x}}\triangleq\sum_{t=1}^T{\mathrm{CRLB}\left( x\left( t \right) \right) /T},
\\
&\mathrm{CRLB}_{\mathrm{angle}}\triangleq \left( \mathrm{CRLB}\left( \gamma _{\mathrm{I}2\mathrm{U}}^{\mathrm{A}} \right) +\mathrm{CRLB}\left( \varphi_{\mathrm{I}2\mathrm{U}}^{\mathrm{A}} \right) \right) /2,
\end{align}
and $\zeta$ denotes the weight coefficient, by adjusting which we can make trade-off between communication performance and localization performance. 
In addition, $\mathrm{CRLB}\left( x\left( t \right) \right) $, $\mathrm{CRLB}\left( \gamma _{\mathrm{I}2\mathrm{U}}^{\mathrm{A}} \right)$, and $\mathrm{CRLB}\left( \varphi _{\mathrm{I}2\mathrm{U}}^{\mathrm{A}} \right) $ denote the CRLBs of $x\left( t \right) $, $\gamma _{\mathrm{I}2\mathrm{U}}^{\mathrm{A}}$, and $ \varphi _{\mathrm{I}2\mathrm{U}}^{\mathrm{A}}$, respectively.
\par By stacking these unknown parameters (i.e., $x\left( t \right) $, $\gamma _{\mathrm{I}2\mathrm{U}}^{\mathrm{A}}$, and $ \varphi _{\mathrm{I}2\mathrm{U}}^{\mathrm{A}}$) to be estimated into a $T+2$ dimensional vector $\bm{\theta}$, we have
\begin{align}
\boldsymbol{\theta }=\left[ x\left( 1 \right) ,\cdots ,x\left( T \right) ,\gamma _{\mathrm{I}2\mathrm{U}}^{\mathrm{A}},\varphi _{\mathrm{I}2\mathrm{U}}^{\mathrm{A}} \right] ^{\mathrm{T}}.
\end{align}
\par Invoking the results in \cite{van2004detection}, we calculate the CRLB of $\left[ \boldsymbol{\theta } \right] _i$ as
\begin{align}
\mathrm{CRLB}\left( \left[ \boldsymbol{\theta } \right] _i \right) =\left[ \mathbf{J}^{-1} \right] _{i,i},i=1,\cdots ,T+2,\label{CRLB}
\end{align}
where $\mathbf{J}\in \mathbb{R} ^{\left( T+2 \right) \times \left( T+2 \right)}$ is the Fisher information matrix (FIM) with respect to $\boldsymbol{\theta }$, whose $(i,j)$-th element is given by
\begin{align}\label{pytheta}
\left[ \mathbf{J} \right] _{i,j}\triangleq-\mathbb{E} \left( \frac{\partial ^2\ln \!\:p\left( \mathbf{y},\boldsymbol{\theta } \right)}{\partial \theta _i\partial \theta _j} \right) ,i,j=1,\cdots ,T+2,
\end{align}
where $p\left( \mathbf{y},\boldsymbol{\theta } \right) $ is the joint probability density function of ${\bf y}$ and ${\bm \theta}$. 
\par Let $p ({\bf x})$ and $p ({\bf y}|{\bf x})$ denote the probability density function of ${\bf x} \triangleq  [x(1),...,x(T)]^\text{T} \in \mathbb{C}^{T \times 1}$ and the conditional probability density function of ${\bf y}$ given ${\bf x}$, respectively.  Then, the joint probability density function $p({\bf y}, {\bm \theta})$ in (\ref{pytheta}) can be expressed as 
\begin{align}
p\left( \mathbf{y},\boldsymbol{\theta } \right) =p\left( \mathbf{x} \right) p\left( \mathbf{y}|\mathbf{x} \right),
\end{align}
based on which, we can decompose $\mathbf{J}$ into two additive parts
\begin{align}
    \mathbf{J}=\mathbf{J}_P+\mathbf{J}_D,\label{J}
\end{align}
where the $(i,j)$-th elements of $\mathbf{J}_P$ and $\mathbf{J}_D$ are respectively defined as
\begin{align}
&[\mathbf{J}_P]_{i,j}\triangleq-\mathbb{E} \left( \frac{\partial ^2\ln \!\:p\left( \mathbf{x} \right)}{\partial \theta _i\partial \theta _j} \right) ,\\
&[\mathbf{J}_D]_{i,j}\triangleq-\mathbb{E} \left( \frac{\partial ^2\ln \!\:p\left( \mathbf{y}|\mathbf{x} \right)}{\partial \theta _i\partial \theta _j} \right) .
\end{align}
\par By calculating $\mathbf{J}_P$ and $\mathbf{J}_D$, we can obtain the FIM $\mathbf{J}$, the CRLB of $\bm \theta$, and the modified CRLB,  as summarized in the following theorem.
\begin{theorem}
\par The FIM with respect to $\boldsymbol{\theta }$ is given by
\begin{align}\label{Jexpress}
\mathbf{J}=\frac{2}{\sigma _{z}^{2}}\left[ \begin{matrix}
	\left\| \boldsymbol{\beta }_x \right\| ^2&		&		&		\mathcal{R} \!\left\{ \boldsymbol{\beta }_{x}^{\mathrm{H}}\boldsymbol{\beta }_{\gamma ,1} \right\}&		\mathcal{R} \!\left\{ \boldsymbol{\beta }_{x}^{\mathrm{H}}\boldsymbol{\beta }_{\varphi ,1} \right\}\\
	&		\ddots&		&		\vdots&		\vdots\\
	&		&		\left\| \boldsymbol{\beta }_x \right\| ^2&		\mathcal{R} \!\left\{ \boldsymbol{\beta }_{x}^{\mathrm{H}}\boldsymbol{\beta }_{\gamma ,T} \right\}&		\mathcal{R} \!\left\{ \boldsymbol{\beta }_{x}^{\mathrm{H}}\boldsymbol{\beta }_{\varphi ,T} \right\}\\
	\mathcal{R} \!\left\{ \boldsymbol{\beta }_{x}^{\mathrm{H}}\boldsymbol{\beta }_{\gamma ,1}\! \right\}&		\cdots&		\mathcal{R} \!\left\{ \boldsymbol{\beta }_{x}^{\mathrm{H}}\boldsymbol{\beta }_{\gamma ,T}\! \right\}&		\sum_{t=1}^T{\!\left\| \boldsymbol{\beta }_{\gamma ,t} \right\| ^2}&		\sum_{t=1}^T{\!\mathcal{R} \!\left\{ \boldsymbol{\beta }_{\gamma ,t}^{\mathrm{H}}\boldsymbol{\beta }_{\varphi ,t}\! \right\}}\\
	\mathcal{R} \!\left\{ \boldsymbol{\beta }_{x}^{\mathrm{H}}\boldsymbol{\beta }_{\varphi ,1}\! \right\}&		\cdots&		\mathcal{R} \!\left\{ \boldsymbol{\beta }_{x}^{\mathrm{H}}\boldsymbol{\beta }_{\varphi ,T}\! \right\}&		\sum_{t=1}^T{\!\mathcal{R} \!\left\{ \boldsymbol{\beta }_{\varphi ,t}^{\mathrm{H}}\boldsymbol{\beta }_{\gamma ,t}\! \right\}}&		\sum_{t=1}^T{\left\| \boldsymbol{\beta }_{\varphi ,t} \right\| ^2\!}\\
\end{matrix} \right] ,
\end{align}
where
\begin{align}
&\boldsymbol{\beta }_x=\mathbf{H}_{\mathrm{I}2\mathrm{U}}\mathbf{\Theta H}_{\mathrm{B}2\mathrm{I}}\mathbf{w},\label{betax}
\\
&\boldsymbol{\beta }_{\gamma ,t} =\mathbf{H}_{\mathrm{I}2\mathrm{U},\gamma}\mathbf{\Theta H}_{\mathrm{B}2\mathrm{I}}\mathbf{w}x\left( t \right) \label{betagammat},
\\
&\boldsymbol{\beta }_{\varphi ,t}=\mathbf{H}_{\mathrm{I}2\mathrm{U},\varphi}\mathbf{\Theta H}_{\mathrm{B}2\mathrm{I}}\mathbf{w}x\left( t \right) \label{betaphit},
\\
&\left[ \mathbf{H}_{\mathrm{I}2\mathrm{U},\varphi} \right] _{m,l}=j\pi \left( m_y\left( m \right) -l_y\left( l \right) \right) \cos \left( \gamma _{\mathrm{I}2\mathrm{U}}^{\mathrm{A}} \right) \cos \left( \varphi _{\mathrm{I}2\mathrm{U}}^{\mathrm{A}} \right) \left[ \mathbf{H}_{\mathrm{I}2\mathrm{U}} \right] _{m,l},
\\
&\left[ \mathbf{H}_{\mathrm{I}2\mathrm{U},\gamma} \right] _{m,l}=j\pi \left( m_z\left( m \right) -l_z\left( l \right) \right) \cos \left( \gamma _{\mathrm{I}2\mathrm{U}}^{\mathrm{A}} \right) \left[ \mathbf{H}_{\mathrm{I}2\mathrm{U}} \right] _{m,l}\notag
\\
&\qquad\qquad\qquad-j\pi \left( m_y\left( m \right) -l_y\left( l \right) \right) \sin \left( \gamma _{\mathrm{I}2\mathrm{U}}^{\mathrm{A}} \right) \sin \left( \varphi _{\mathrm{I}2\mathrm{U}}^{\mathrm{A}} \right) \left[ \mathbf{H}_{\mathrm{I}2\mathrm{U}} \right] _{m,l}.
\end{align}
\par The CRLB of $x\left( t \right) $, $\gamma _{\mathrm{I}2\mathrm{U}}^{\mathrm{A}}$, and $\varphi _{\mathrm{I}2\mathrm{U}}^{\mathrm{A}}$ are respectively given by
\begin{align}
&\mathrm{CRLB}\left( x\left( t \right) \right) =\left[ \mathbf{J}^{-1} \right] _{t,t},t=1,\cdots ,T,
\\
&\mathrm{CRLB}\left( \gamma _{\mathrm{I}2\mathrm{U}}^{\mathrm{A}} \right) =\left[ \mathbf{J}^{-1} \right] _{T+1,T+1},
\\
&\mathrm{CRLB}\left( \varphi _{\mathrm{I}2\mathrm{U}}^{\mathrm{A}} \right) =\left[ \mathbf{J}^{-1} \right] _{T+2,T+2}.
\end{align}

The modified CRLB for the IRS-aided NO-ISAC system is given by
\begin{align}
\mathrm{CRLB}_{\mathrm{ISAC}}\triangleq \zeta \lg \left( \sum_{t=1}^T{\left[ \mathbf{J}^{-1} \right] _{t,t}/T} \right) +\left( 1-\zeta \right) \lg \left( \sum_{t=T+1}^{T+2}{\left[ \mathbf{J}^{-1} \right] _{t,t}/2} \right).
\end{align}
\end{theorem}
\begin{proof}
See Appendix A.
\end{proof}
% \par Hence, with the new performance metric for communication, i.e., the communication CRLB $\mathrm{CRLB}(x(t))$, the performance of the IRS-aided NO-ISAC system can be characterized by a unified metric, i.e., CRLB.
\par Let $I(t)$ denote the mutual information between $x(t)$ and its estimator $\hat{x}(t)$. The relationship between the communication CRLB  $\mathrm{CRLB}(x(t))$ and the mutual information  $I(t)$ is provided by the following corollary.
\begin{corollary}
The mutual information $I(t)$ is given by
\begin{align}
I\left( t \right) =\frac{1}{2}\log \left( \frac{\sigma _{x}^{2}}{\mathrm{CRLB}\left( x\left( t \right) \right)} \right) .
\end{align}
\end{corollary}
\begin{proof}
The mutual information between $x(t)$ and $\hat{x}(t)$ can be expressed as
\begin{align}
I\left( t \right) =H\left( x\left( t \right) \right) -H\left( x\left( t \right) |\hat{x}\left( t \right) \right),
\end{align}
where $H\left( x\left( t \right) \right) $ denotes the entropy of $x\left( t \right) $ and $H\left( x\left( t \right) |\hat{x}\left( t \right) \right) $ denotes the conditional entropy of $x\left( t \right) $ given the estimated signal $\hat{x}\left( t \right) $. Note that, with the given estimated signal $\hat{x}\left( t \right) $, the conditional probability distribution of $x\left( t \right) $ can be expressed as
\begin{align}
\left( x\left( t \right) |\hat{x}\left( t \right) \right) \sim \mathcal{C} \mathcal{N} \left( \hat{x}\left( t \right) ,\mathrm{CRLB}\left( x\left( t \right) \right) \right).
\end{align}
According to \cite{AdaliComplex}, the entropy of a complex random variable is defined as the entropy of its real composite. Therefore, the mutual information $I(t)$ can be calculated as
\begin{align}
I(t)&=H\left( \mathcal{R} \left\{ x\left( t \right) \right\} \right) -H\left( \mathcal{R} \left\{ x\left( t \right) |\hat{x}\left( t \right) \right\} \right) 
\\
&=\frac{1}{2}\log \pi e\sigma _{x}^{2}-\frac{1}{2}\log \pi e\mathrm{CRLB}\left( x\left( t \right) \right)\notag 
\\
&=\frac{1}{2}\log \left( \frac{\sigma _{x}^{2}}{\mathrm{CRLB}\left( x\left( t \right) \right)} \right). \notag
\end{align}
\end{proof}

\section{CRLB-based Beamforming Design}\label{section4}
In this section, by exploiting the modified CRLB, we propose a joint design of active beamforming at the BS and passive beamforming at the IRS, with both communication and localization performances taken into account.
% To measure the joint performance of communication and localization, we define the modified CRLB (in dB) of the NO-ISAC system as 
% \begin{align}\label{FEMSE}
% \mathrm{CRLB}_\text{ISAC}\triangleq \zeta \lg \left( \mathrm{CRLB}_{\mathrm{x}} \right) +\left( 1-\zeta \right) \lg \left( \mathrm{CRLB}_{\mathrm{angle}} \right) ,
% \end{align}
% where
% \begin{align}
% &\mathrm{CRLB}_{\mathrm{x}}=\sum_{t=1}^T{\mathrm{CRLB}\left( x\left( t \right) \right) /T},
% \\
% &\mathrm{CRLB}_{\mathrm{angle}}= \left( \mathrm{CRLB}\left( \gamma _{\mathrm{I}2\mathrm{U}}^{\mathrm{A}} \right) +\mathrm{CRLB}\left( \varphi_{\mathrm{I}2\mathrm{U}}^{\mathrm{A}} \right) \right) /2,
% \end{align}
% and $\zeta$ denotes the weight coefficient, by adjusting which we can make trade-off between communication performance and localization performance.
\par As such, we formulate the CRLB minimization problem as
\begin{subequations}
\begin{align}
\text {(P1)}:\min_{\mathbf{w},\boldsymbol{\xi }}&\quad \mathrm{CRLB}_{\mathrm{ISAC}} \\
\text{s.t.}&\quad\left\| \mathbf{w} \right\| ^2\leqslant P_t,\\
&\quad\,\,\vartheta _l\in \mathcal{F} ,\forall l=1,\cdots ,L.
\end{align}
\end{subequations}
By invoking the results in Theorem 1, the objective function can be expressed as 
\begin{align}
\mathrm{CRLB}_{\mathrm{ISAC}}&=\zeta \lg \left( \frac{\sum_{t=1}^T{\left[ \mathbf{J}_{\xi}^{-1} \right] _{t,t}}}{T\left| \mathbf{a}_{\mathrm{B}}^{\mathrm{H}}\left( \gamma _{\mathrm{B}2\mathrm{I}}^{\mathrm{D}} \right) \mathbf{w} \right|^2} \right) +\left( 1-\zeta \right) \lg \left( \frac{\sum_{t=T+1}^{T+2}{\left[ \mathbf{J}_{\xi}^{-1} \right] _{t,t}}}{2\left| \mathbf{a}_{\mathrm{B}}^{\mathrm{H}}\left( \gamma _{\mathrm{B}2\mathrm{I}}^{\mathrm{D}} \right) \mathbf{w} \right|^2} \right) 
\\
&=V_{\xi}-2\lg \left| \mathbf{a}_{\mathrm{B}}^{\mathrm{H}}\left( \gamma _{\mathrm{B}2\mathrm{I}}^{\mathrm{D}} \right) \mathbf{w} \right|,\notag
\end{align}
where
\begin{align}
&V_{\xi}=\lg \left( \sum_{t=1}^T{\left[ \mathbf{J}_{\xi}^{-1} \right] _{t,t}}/T \right) +\left( 1-\zeta \right) \lg \left( \sum_{t=T+1}^{T+2}{\left[ \mathbf{J}_{\xi}^{-1} \right] _{t,t}}/2 \right),\label{V}\\
&\mathbf{J}_{\boldsymbol{\xi }}=\frac{\mathbf{J}}{\left| \mathbf{a}_{\mathrm{B}}^{\mathrm{H}}\left( \gamma _{\mathrm{B}2\mathrm{I}}^{\mathrm{D}} \right) \mathbf{w} \right|^2}.
\end{align}
Noticing that $V_{\xi}$ is independent of ${\bf w}$, we can decompose the problem (P1) into two uncoupled subproblems with
respect to ${\bf w}$ and ${\bm \xi}$, respectively.

\subsection{BS Active Beamforming}
The subproblem with respect to the BS beamforming vector ${\bf w}$ can be formulated as
\begin{subequations} 
\begin{align}
\text {(P2)}:\max_{\mathbf{w}} &\quad \left| \mathbf{a}_{\mathrm{B}}^{\mathrm{H}}\left( \gamma _{\mathrm{B}2\mathrm{I}}^{\mathrm{D}} \right) \mathbf{w} \right|^2 \\
\text{s.t.}&\quad\left\| \mathbf{w} \right\| ^2\leqslant P_t.
\end{align}
\end{subequations}
It can be easily verified that the optimal solution is
\begin{align}\label{wopt}
\mathbf{w}=\sqrt{\frac{P_t}{N_t}}\mathbf{a}_{\mathrm{B}}\left( \gamma _{\mathrm{B}2\mathrm{I}}^{\mathrm{D}} \right).
\end{align}
\subsection{IRS Passive Beamforming}
The optimization subproblem with respect to the IRS phase shift beam ${\bm \xi}$ can be formulated as
\begin{align}
\text {(P3)}:\min_{\boldsymbol{\xi }}&\quad V_{\xi} \\
\text{s.t.}&\quad\vartheta _l\in \mathcal{F},\forall l=1,\cdots ,L.
\end{align}
\par Due to the complicated expression of the objective function as well as the non-convex phase shift constraints, the problem (P3) is difficult to solve. Responding to this, we propose a CE-based method to optimize $\boldsymbol{\xi }$.
\par First, we set the probability matrix corresponding to $\boldsymbol{\xi }$ as $\mathbf{P}=\left[ \mathbf{p}_1,\cdots ,\mathbf{p}_l,\cdots ,\mathbf{p}_L \right] \in \mathbb{C} ^{2^b\times L}$, where $\mathbf{p}_l=\left[ p_{l,1},\cdots ,p_{l,s},\cdots ,p_{l,2^b} \right] ^T$
denotes the probability parameter for $\vartheta _l$, with its entry $p_{l,s}$ satisfying the probability constraints $0\leqslant p_{l,s}\leqslant 1$ and $\sum_{s=1}^{2^b}{p_{l,s}}=1$. Subsequently, we consider that $\vartheta _l$ takes a value from $\mathcal{F} $ with an equal probability at first, and initialize the  probability matrix as $\mathbf{P}^0=\frac{1}{2^b}\times \mathbf{1}_{2^b\times L}$. Then, in the $i$-th iteration, we randomly generate $C$ candidates $\left\{ \boldsymbol{\xi }^c \right\} _{c=1}^{C}$ according to the probability distribution function given by
\begin{align} 
\Xi \left( \boldsymbol{\xi };\mathbf{P}^i \right) =\prod_{l=1}^L{\left( \left( \prod_{s=1}^{2^b-1}{\left( p_{l,s}^{i} \right) ^{\varGamma \left( \vartheta _l,\mathcal{F} \left( s \right) \right)}} \right) \times \left( 1-\prod_{s=1}^{2^b-1}{\left( p_{l,s}^{i} \right) ^{\varGamma \left( \vartheta _l,\mathcal{F} \left( s \right) \right)}} \right) \right)},
\end{align}
where $\mathcal{F} \left( s \right) $ denotes the $c$-th entry of $\mathcal{F}$, and $\varGamma \left( \vartheta _l,\mathcal{F} \left( s \right) \right) $ is a judge function  given by
\begin{align} 
\varGamma \left( \vartheta _l,\mathcal{F} \left( s \right) \right) =\begin{cases}
	1,\vartheta _l=\mathcal{F} \left( s \right)\\
	0,\vartheta _l\ne \mathcal{F} \left( s \right)\\
\end{cases}.
\end{align}
For each phase shift beam candidate $\boldsymbol{\xi }^c$, we calculate the corresponding $V_{\xi}\left( \boldsymbol{\xi }^c \right) $ by (\ref{V}).
\par After that, we sort $\mathcal{V} \!=\!\left\{ V_{\xi}\left( \boldsymbol{\xi }^c \right) \right\} _{c=1}^{C}$ in ascending order and select $C_{\mathrm{elite}}$ phase shift beam samples corresponding to the $C_{\mathrm{elite}}$ smallest modified CRLBs in $\mathcal{V} $, i.e., $\boldsymbol{\xi }^{c_q},c_q=1,\cdots ,C_{\mathrm{elite}}$. Based on the selected $C_{\mathrm{elite}}$ samples, we update the probability matrix in the $(i+1)$-th iteration as $\mathbf{P}^{i+1}$ with its elements given by
\begin{align} \label{BF2}
p_{l,s}^{i+1}=\frac{1}{C_{\mathrm{elite}}}\sum_{c_q=1}^{C_{\mathrm{elite}}}{\varGamma \left( \vartheta _{l}^{c_q},\mathcal{F} \left( s \right) \right)},l=1,\cdots ,L,\,s=1,\cdots ,2^b.
\end{align}
\par Repeat the above process until the difference between the maximum and minimum of $\mathcal{V} $ is less than the threshold $\kappa$, which indicates that the probability matrix corresponding to $\boldsymbol{\xi }$ is stable. The procedure of the beamforming algorithm is summarized in Algorithm 1.
\begin{remark}
The computational complexity of Algorithm 1 is mainly composed of the number of generated candidates per iteration $C$ and the calculation of the FIM matrix $\mathbf{J}$, where the complexity of calculating $\mathbf{J}$ mainly comes from $\boldsymbol{\beta }_x$ in (\ref{betax}), $\boldsymbol{\beta }_{\gamma ,t}$ in (\ref{betagammat}) and $\boldsymbol{\beta }_{\varphi ,t}$ in (\ref{betaphit}). As such, the complexity of Algorithm 1 is $\mathcal{O} \left( L^3\left( 2T+1 \right) C \right) $.
\end{remark}
\begin{algorithm}[htp]
\caption{CE-based Beamforming Algorithm}
\hspace*{0.02in}  
\begin{algorithmic}[1]
\State Initialize $\mathcal{F} =\left\{ 0,\frac{2\pi}{2^b},\cdots,\frac{2\pi}{2^b}\left( 2^b-1 \right) \right\} $, $\mathbf{P}^0=\frac{1}{2^b}\times \mathbf{1}_{2^b\times L}$,
 and $i=1$.
\State Calculate $\mathbf{w}$ based on (\ref{wopt}).
\Repeat
    \State Randomly generate $C$ candidates $\left\{ \boldsymbol{\xi }^c \right\} _{c=1}^{C}$ based on $\Xi \left( \boldsymbol{\xi };\mathbf{P}^i \right)$.
    \State Calculate the modified mean square error $V_{\xi}\left( \boldsymbol{\xi }^c \right) $ by (\ref{V}).
    \State Sort $\mathcal{V} \!=\!\left\{ V_{\xi}\left( \boldsymbol{\xi }^c \right) \right\} _{c=1}^{C}$ in ascending order.
    \State Select $C_{\mathrm{elite}}$ phase shift beam samples corresponding to the $C_{\mathrm{elite}}$ smallest modified CRLBs in $\mathcal{V} $.
    \State Update $\mathbf{P}^{i+1}$ according to (\ref{BF2}).
    \State Let $i \rightarrow i+1 $.
\Until $\left| \max \left( \mathcal{V}  \right) -\min \left( \mathcal{V}  \right) \right|<\kappa $
\end{algorithmic}
\hspace*{0.02in} {\bf Output:}
$\boldsymbol{\xi }^{\mathrm{opt}}=\boldsymbol{\xi }^1$,  $\mathbf{w}^{\mathrm{opt}}=\mathbf{w}$.
\end{algorithm}

\section{Numerical Results}\label{section5}
This section presents numerical results to investigate the performance of the proposed IRS-aided NO-ISAC system as well as to verify the effectiveness of the proposed beamforming algorithm. The simulation setup is shown in Fig.~\ref{setup}, where the user is on the horizontal floor, the BS is $10$ meters (m) above the horizontal floor, and the IRS is $5$ m above the horizontal floor. The distances from the BS to the IRS and from the IRS to the user are set to be $d_{\mathrm{B}2\mathrm{I}}=30$ m and $d_{\mathrm{I}2\mathrm{U}}=10 $ m, respectively. The path loss exponents from the BS to the IRS and from the IRS to the user are set as $2.3$ and $2.2$, respectively. The path loss at the reference distance of $1$ m is set as 30 dB. The received SNR at the user is defined as $SNR\triangleq \frac{\left| \alpha _{\mathrm{B}2\mathrm{I}}\alpha _{\mathrm{I}2\mathrm{U}} \right|^2}{\sigma _{z}^{2}}$. Unless otherwise specified, the following setups are adopted: $SNR=0$ dB, $N_{t}=8$, $L=4\times4$, $M=4\times4$, $P_t=1$, $T=2$, $\zeta=0.5$, $b=2$, $C=5\times L=80$, $C_{\mathrm{elite}}=C/10=8$, and $\kappa =10^{-3}$.
\begin{figure}[htb]
  \centering
  \includegraphics[width=4in]{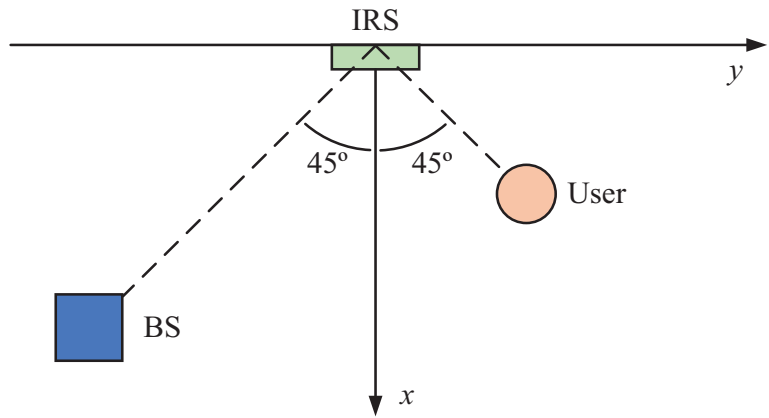}
  \caption{Simulation setup (top view).}
  \label{setup}
\end{figure}

\subsection{IRS-Aided NO-ISAC System vs. IRS-Aided Localization System}
\par In this subsection, we compare the performance of the IRS-aided NO-ISAC system with that of the IRS-aided localization system. For the former, the communication signals sent by the BS are used for both data transmission and localization, while for the latter, the positioning reference signals sent by the BS are used only for localization and known to the user. We adopt the CRLB for angle estimation to evaluate the localization performance, and based on Corollary 1, the communication performance of the IRS-aided NO-ISAC system is measured by the average mutual information defined as
\begin{align}
I_{\mathrm{NO}-\mathrm{ISAC}}\triangleq \frac{1}{T}\sum_{t=1}^T{I\left( t \right)}\label{IISAC},
\end{align}
\begin{figure}[htb]
  \centering
  \includegraphics[width=4in]{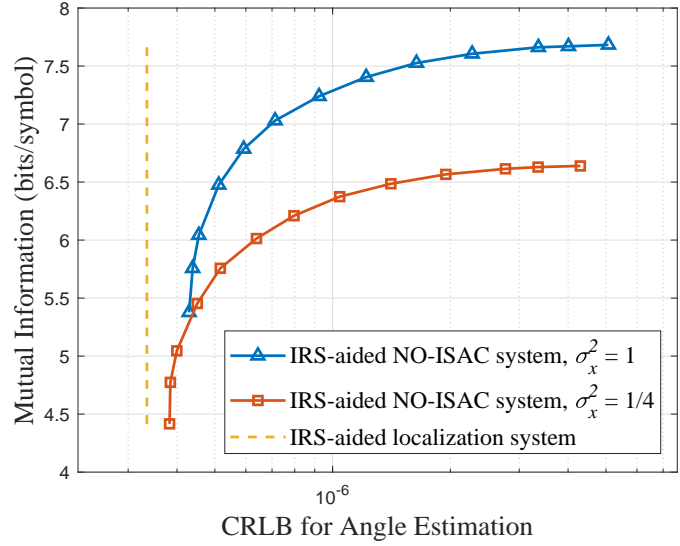}
  \caption{Performance comparison of the NO-ISAC system and the localization system.}
  \label{IvsL}
\end{figure}
\par Fig.~\ref{IvsL} compares the performance of the IRS-aided NO-ISAC system and the IRS-aided localization system with different variances of $x(t)$, where we set $T=5$. The localization performance of the IRS-aided NO-ISAC system is worse than that of the IRS-aided localization system, due to the randomness of the communication signal and the requirement of allocating partial spatial resources for communication to achieve simultaneous communication and localization. Despite the loss of location sensing performance, the IRS-aided NO-ISAC system can realize concurrent communication and localization, while the IRS-aided localization system only has the location sensing function. 
For example, although the location sensing accuracy of the IRS-aided NO-ISAC system with $\sigma_x^2=1$ is $2$ times lower than that of the IRS-aided localization system, the communication mutual information of $7$ bits/symbol can be achieved in the IRS-aided NO-ISAC system. Moreover, we can observe that increasing the variance of the communication signal is adverse for localization while beneficial for communication.
\begin{figure}[htb]
  \centering
  \includegraphics[width=4in]{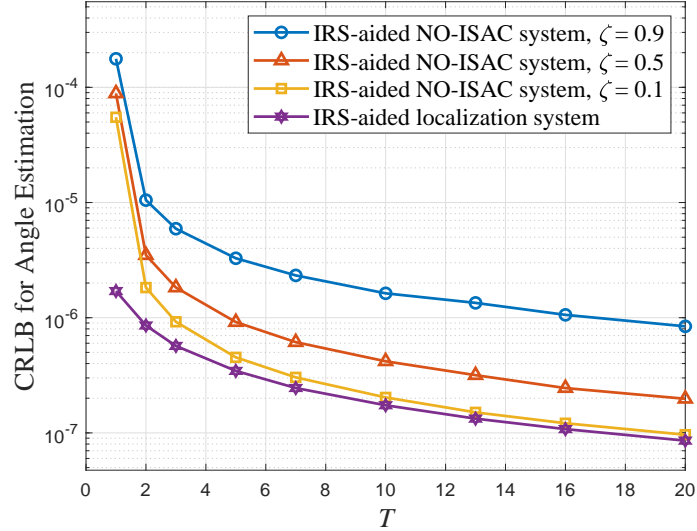}
  \caption{CRLB for angle estimation versus $T$.}
  \label{sigma_T}
\end{figure}
\par Fig.~\ref{sigma_T} presents the CRLB for angle estimation versus the number of time slots (per time block) $T$ with different weight coefficient $\zeta$. It is obvious that the localization performance gap between the two systems becomes smaller, as the weight coefficient $\zeta$ decreases, indicating that allocating more spatial resources for localization could effectively improve the location sensing accuracy of the IRS-aided NO-ISAC system. Moreover, when allocating sufficient spatial resources for localization (i.e., $\zeta$ is enough small), the performance gap between the two systems gradually vanishes with the increase of the number of time slots (per time block) $T$. This is because collecting more data for localization helps suppress the adverse effect of the randomness caused by the communication signal.
\begin{figure}[htb]
  \centering
  \includegraphics[width=4in]{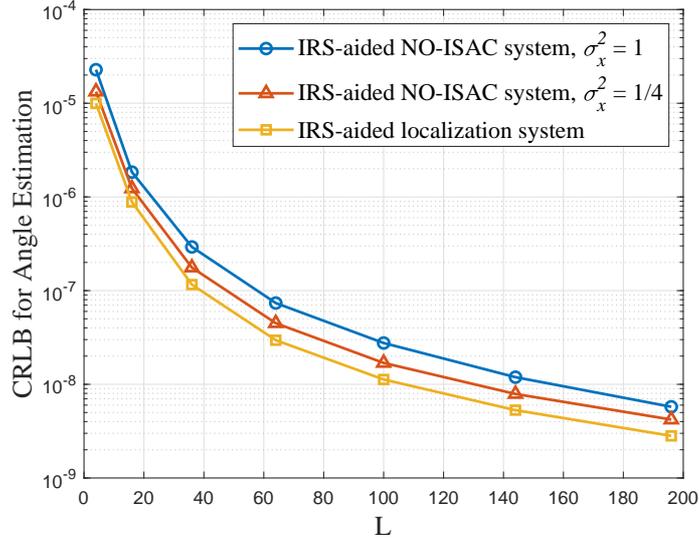}
  \caption{CRLB for angle estimation versus $L$.}
  \label{sigma_L}
\end{figure}
\par Fig.~\ref{sigma_L} shows the CRLB for angle estimation versus the number of IRS elements $L$ with different variances of $x(t)$, where we set $\zeta=0.1$. As can be seen, the localization performance of  both the IRS-aided NO-ISAC system and the IRS-aided localization system improves significantly as the number of IRS elements increases, especially in the case of a small number of IRS elements, demonstrating the advantage of using the IRS for localization.
\subsection{IRS-Aided NO-ISAC System vs. IRS-Aided TD-ISAC System}
\begin{figure}[htb]
  \centering
  \includegraphics[width=5in]{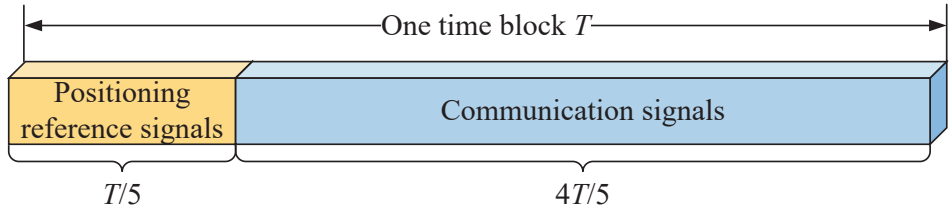}
  \caption{Transmission protocol of the IRS-aided TD-ISAC system in each time block.}
   \label{SSAC_simu}
\end{figure}
\par In this subsection, we compare the performance of the IRS-aided NO-ISAC system and the IRS-aided TD-ISAC system.
For the IRS-aided NO-ISAC system, localization and data transmission are conducted simultaneously by sending communication signals from the BS to the user during the whole time block. For the IRS-aided TD-ISAC system, as illustrated in Fig.~\ref{SSAC_simu}, localization is conducted by sending positioning reference signals from the BS to the user in the first $T/5$ time slots of each time block, while data transmission is conducted by sending communication signals in the remaining  time slots. Therefore, for communication performance comparison, the average mutual information of the IRS-aided TD-ISAC system is defined as
\begin{align}
&I_{\mathrm{TD-ISAC}}\triangleq\frac{1}{T}\sum_{t=\frac{T}{5}}^T{I\left( t \right)}.
\end{align}
\begin{figure}[htb]
  \centering
  \subfigure[Communication performance.]
  {
  \label{IVSD_T_C}
  \includegraphics[width=3.5in]{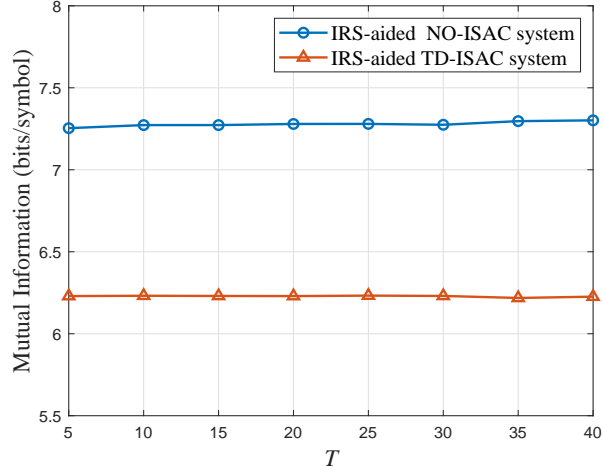}
  }
  \subfigure[Localization performance.]
  {
  \label{IVSD_T_L}
  \includegraphics[width=3.5in]{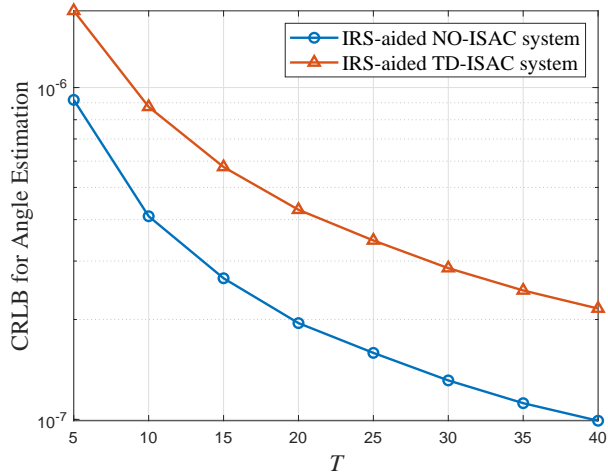}
  }
  \caption{Performance comparison of the NO-ISAC system and the TD-ISAC system versus $T$.}
\end{figure} 
\par Fig.~\ref{IVSD_T_C} and Fig.~\ref{IVSD_T_L} compare the performance of the IRS-aided NO-ISAC system with that of the IRS-aided TD-ISAC system versus the number of time slots (per time block) $T$. It is obvious that both the communication performance and the localization performance of the IRS-aided NO-ISAC system are superior to those of the IRS-aided TD-ISAC system, demonstrating the advantage of the proposed IRS-aided NO-ISAC system. This superiority is due to the following two reasons. On the one hand, the IRS-aided NO-ISAC system conducts data transmission and localization simultaneously within the whole time block. On the other hand, for the IRS-aided NO-ISAC system, the adverse effect of signal randomness on localization performance can be effectively eliminated by increasing $T$. As the number of time slots (per time block) becomes larger, the advantage of the IRS-aided NO-ISAC system in terms of localization performance becomes more significant.
In addition, since the number of communication symbols to be estimated per time block is proportional to $T$ and these symbols are independent of each other, the communication performance of two systems remains stable as the number of time slots (per time block) increases.
\begin{figure}[htb]
  \centering
  \subfigure[Communication performance.]
  {
  \label{IVSD_SNR_C}
  \includegraphics[width=3.5in]{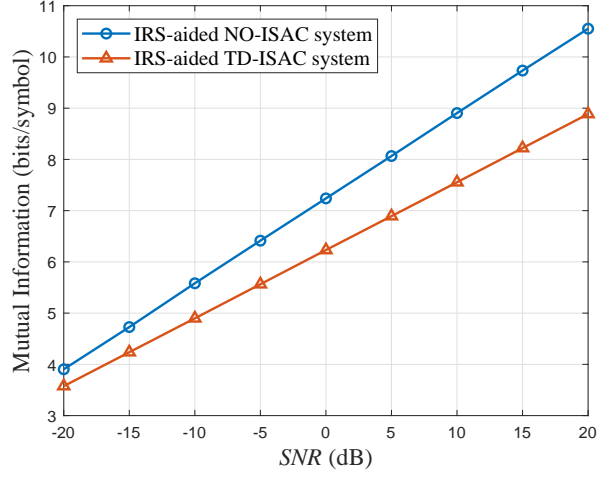}
  }
  \subfigure[Localization performance.]
  {
  \label{IVSD_SNR_L}
  \includegraphics[width=3.5in]{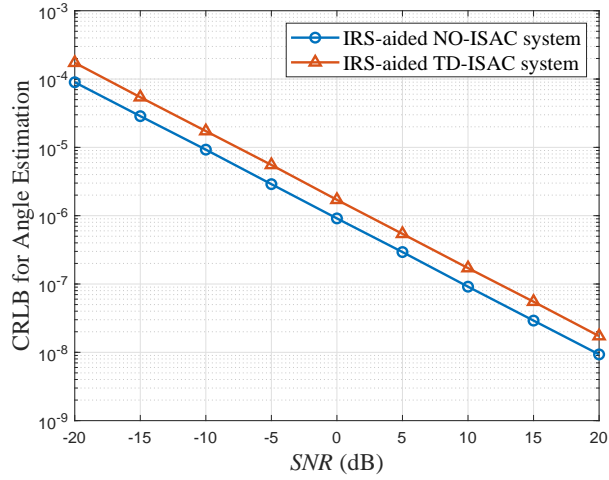}
  }
  \caption{Performance comparison of the NO-ISAC system and the TD-ISAC system versus SNR.}
\end{figure} 
\par Fig.~\ref{IVSD_SNR_C} and Fig.~\ref{IVSD_SNR_L} compare the performance of the IRS-aided NO-ISAC system and the IRS-aided TD-ISAC system versus the user received SNR. In the whole SNR regime, the IRS-aided NO-ISAC system outperforms the IRS-aided TD-ISAC system in terms of both communication performance and localization performance. Moreover, as the SNR increases, the communication performance gap becomes more significant, which indicates that the IRS-aided NO-ISAC system can benefit more from the increase of the user received SNR. This communication performance gap is mainly because, the IRS-aided NO-ISAC system carries out data transmission during the whole time block, while the IRS-aided TD-ISAC system carries out data transmission only within partial time slots of a time block.
\begin{figure}[htb]
  \centering
  \subfigure[Communication performance.]
  {
  \label{IVSD_L_C}
  \includegraphics[width=3.5in]{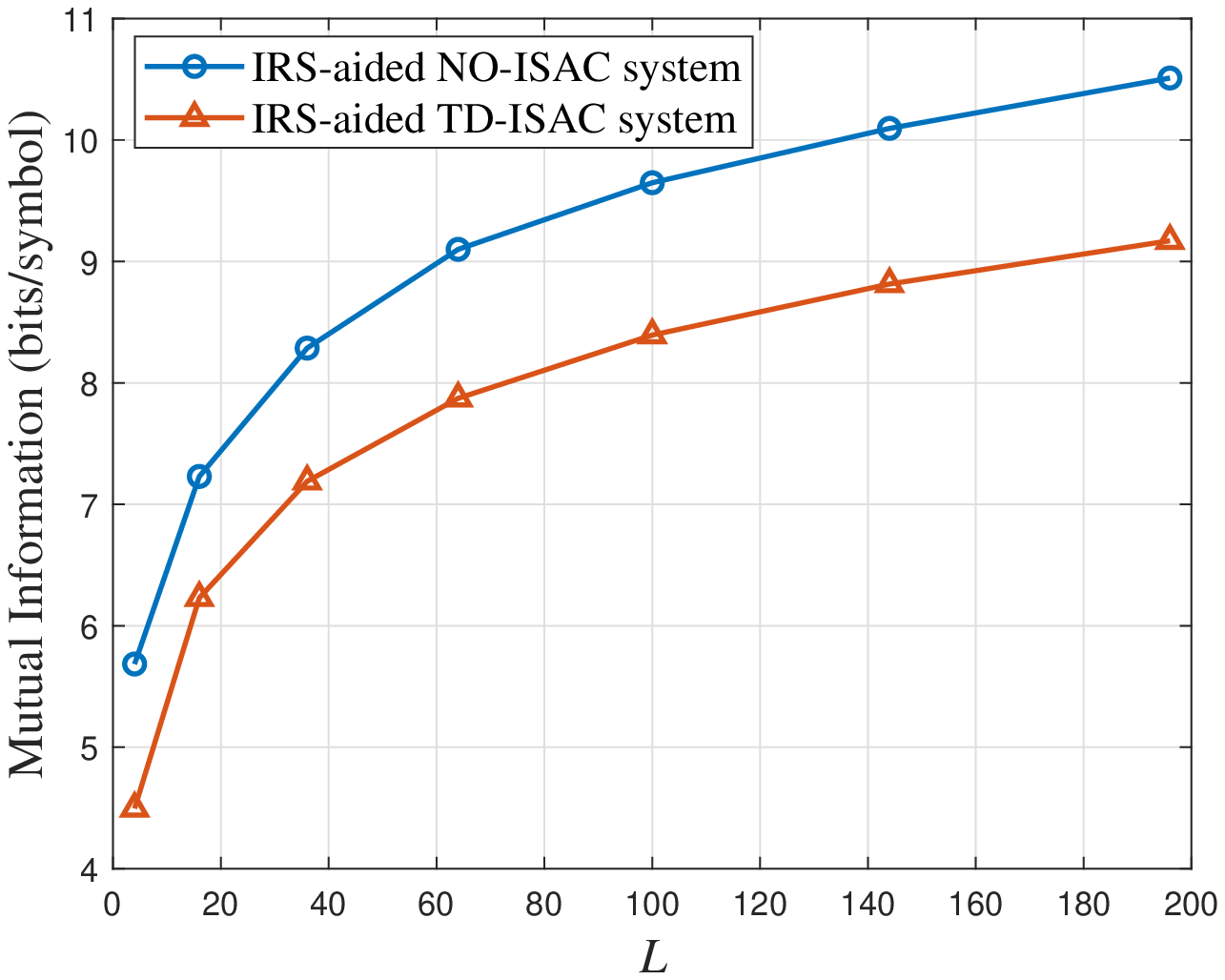}
  }
  \subfigure[Localization performance.]
  {
  \label{IVSD_L_L}
  \includegraphics[width=3.5in]{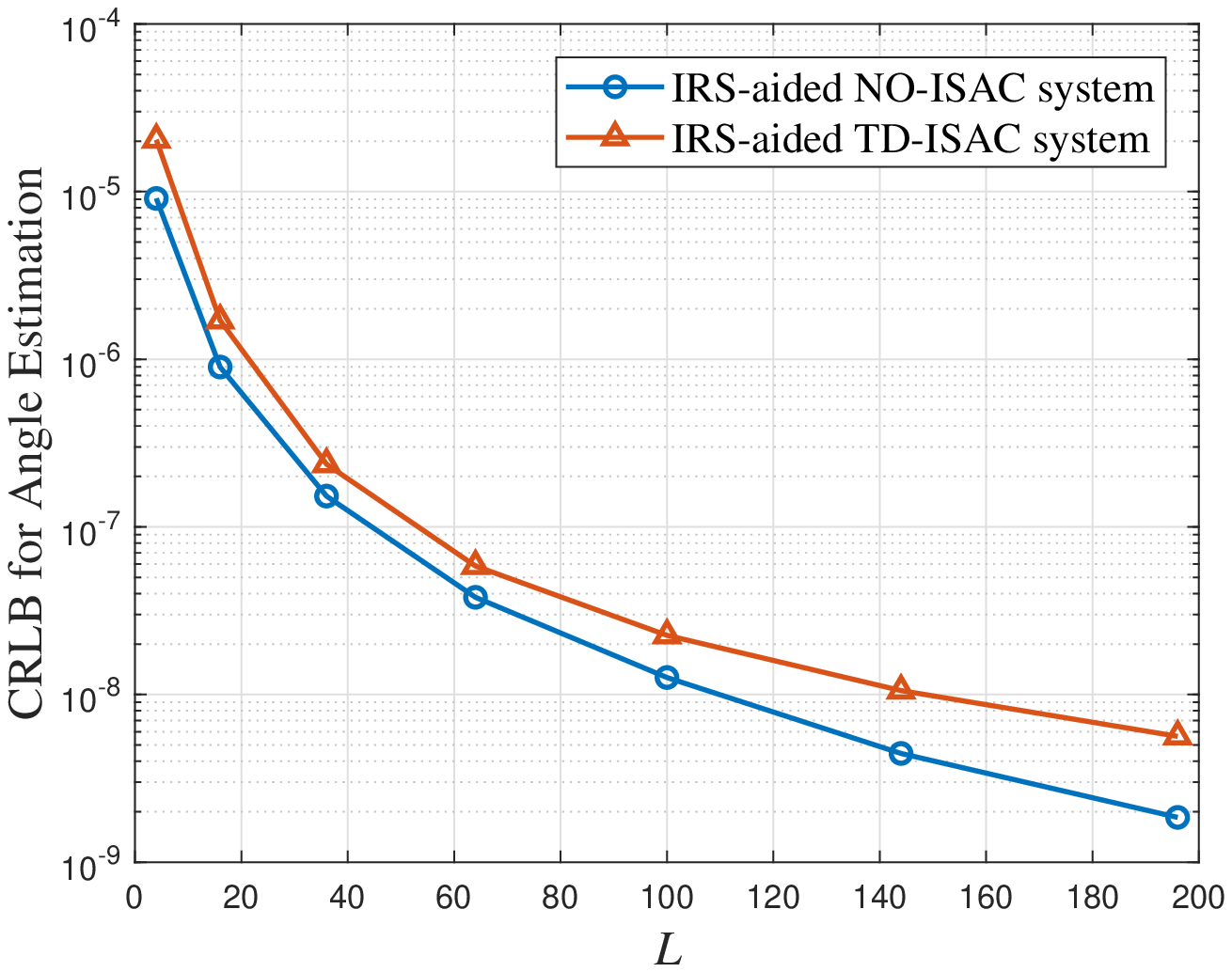}
  }
  \caption{Performance comparison of the NO-ISAC system and the TD-ISAC system versus $L$.}
\end{figure} 
\par Finally, in Fig.~\ref{IVSD_L_C} and Fig.~\ref{IVSD_L_L}, we compare the performance of the IRS-aided NO-ISAC system and the IRS-aided TD-ISAC system versus the number of IRS elements. For any configuration of the number of IRS elements, the IRS-aided NO-ISAC system performs better than the IRS-aided TD-ISAC system in terms of both communication performance and localization performance. As the number of IRS elements increases, the communication performance and localization performance of both systems improve remarkably
due to the increased IRS beamforming gain, and the advantage of the IRS-aided NO-ISAC system over the  IRS-aided TD-ISAC system becomes more pronounced.

\subsection{Trade-off between Communication Performance and Localization Performance}
In this subsection, we demonstrate the effectiveness of the proposed CRLB-based beamforming algorithm and reveal the trade-off between communication performance and localization performance. We adopt the average mutual information in (\ref{IISAC}) and the CRLB for angle estimation to measure the communication performance and localization performance of the IRS-aided NO-ISAC system, respectively.
\begin{figure}[htb]
  \centering
  \includegraphics[width=4in]{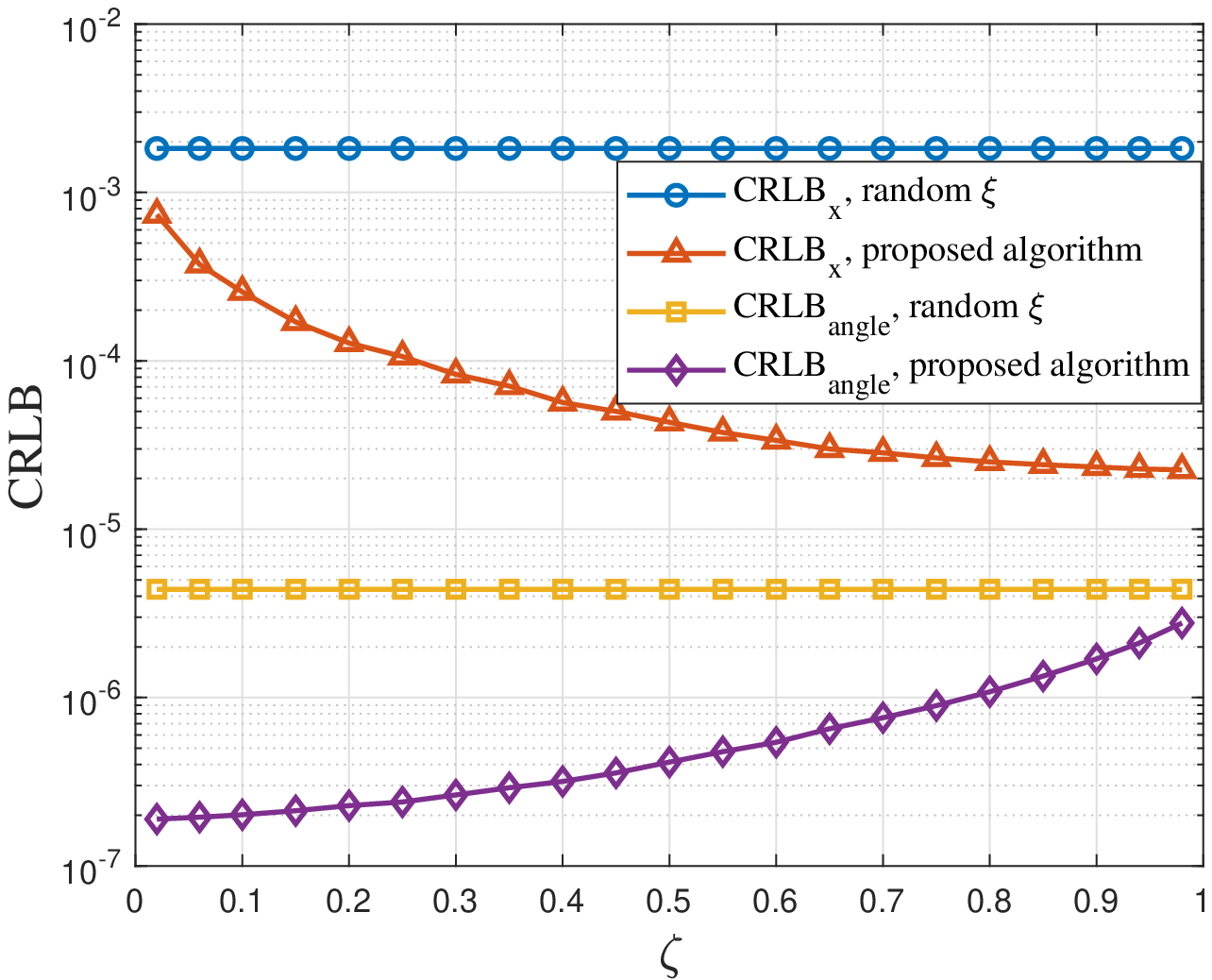}
  \caption{Beamforming performance versus $\zeta$.}
  \label{BF}
\end{figure}
\par Fig.~\ref{BF} illustrates the performance of the proposed CRLB-based beamforming algorithm versus the weight coefficient $\zeta$. For comparison, the performance of the beamforming algorithm with random phase shift beam ${\bm\xi }$ is presented as a benchmark. It is obvious that the proposed beamforming algorithm achieves much better communication performance and localization performance than the random phase shift algorithm. In addition, with the increase of the weight coefficient $\zeta$, the communication CRLB decreases while the CRLB for angle estimation increases, demonstrating that the proposed CRLB-based beamforming algorithm can effectively make trade-off between the communication performance and the localization performance by adjusting the weight coefficient $\zeta$. This result also proves the effectiveness of the proposed modified CRLB for guiding the design of the IRS-aided NO-ISAC system as a unified performance metric.

\begin{figure}[htb]
  \centering
  \includegraphics[width=4in]{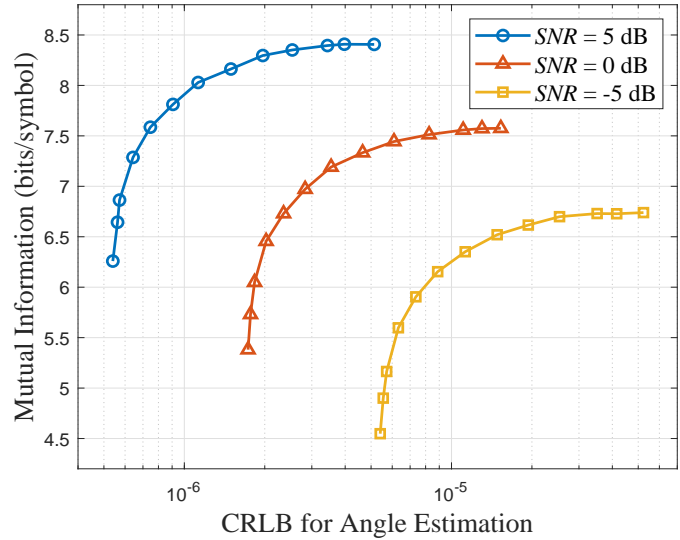}
  \caption{Trade-off between communication performance and localization performance with different SNRs.}
  \label{tradeoff_SNR}
\end{figure}
\par Fig.~\ref{tradeoff_SNR} shows the trade-off between communication performance and localization performance with different user received SNRs. We can observe that increasing the communication/localization performance would degrade the localization/communication performance, which reveals the displacement relation between  communication and localization performance. By sacrificing the performance of one, the performance of another can be improved. Moreover, the communication-localization curve includes three regions, trade-off region, communication saturation region as well as localization saturation region. In the trade-off region, sacrificing the performance of one can effectively enhance that of another. For example, when $SNR=0$ dB, by sacrificing the mutual information from $7.5$ bits/symbol to $6.5$ bits/symbol, the CRLB for angle estimation improves from $6\times 10^{-6}$ to $2\times 10^{-6}$. In the communication/localization saturation region, despite sacrificing the performance of one a lot, little performance gain of another can be obtained. For example, when $SNR=5$ dB, the communication performance only improves from $8.3$ bits/symbol to $8.4$ bits/symbol despite sacrificing the CRLB for angle estimation from $2\times 10^{-6}$ to $5\times 10^{-6}$.
\begin{figure}[htb]
  \centering
  \includegraphics[width=4in]{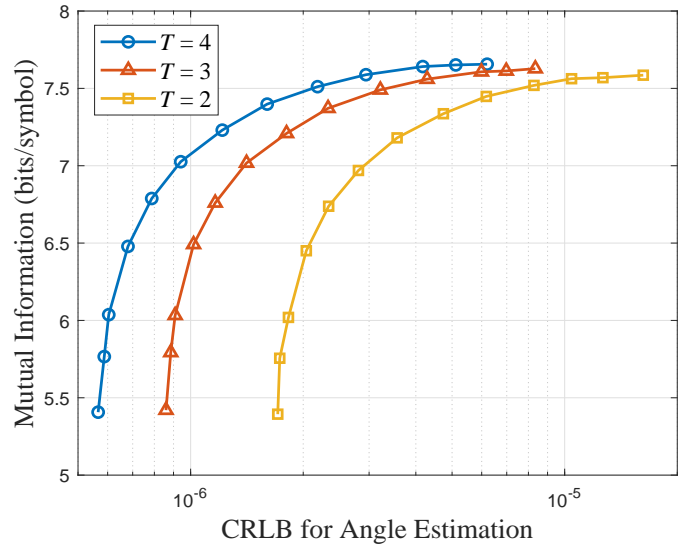}
  \caption{Trade-off between communication performance and localization performance with different $T$.}
  \label{tradeoff_T}
\end{figure}
\par Fig.~\ref{tradeoff_T} presents the trade-off between communication performance and localization performance with different numbers of time slots (per time block) $T$. It can be seen that increasing the number of time slots (per time block) drastically improves the localization performance but has little effect on the communication performance. This indicates that the localization performance sacrificed for improving the communication performance could be compensated by increasing the number of time slots (per time block). For example, as the mutual information increases from $6.5$ bits/symbol to $7.5$ bits/symbol, the CRLB for angle estimation remains constant at $2\times 10^{-6}$ by increasing $T$ from $2$ to $4$.
\begin{figure}[htb]
  \centering
  \includegraphics[width=4in]{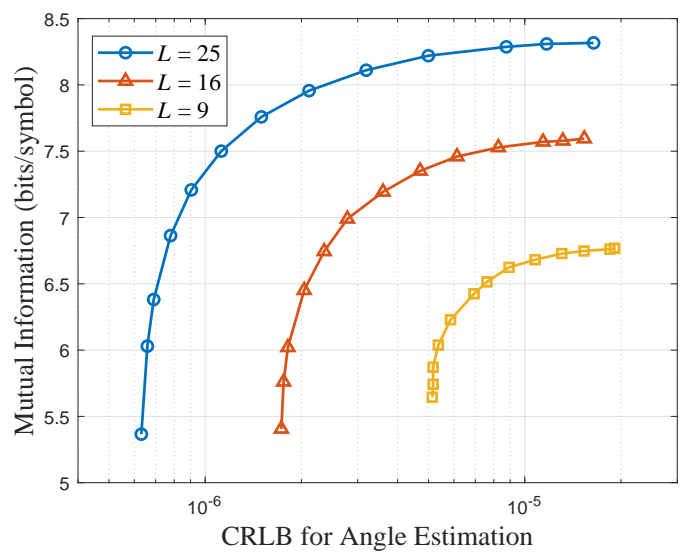}
  \caption{Trade-off between communication performance and localization performance with different $L$.}
  \label{tradeoff_L}
\end{figure}
\par Finally, Fig.~\ref{tradeoff_L} illustrates the trade-off between communication and localization performance with different numbers of IRS elements $L$. It can be readily seen that the adjustable range of both communication and localization performance is enlarged considerably as the number of IRS elements increases. The reason is that more IRS elements would bring more spatial resources, which are flexibly allocated by the proposed beamforming algorithm for balancing communication performance and localization performance. For example, as the number of IRS elements increases from
9 to 25, the adjustable range of the mutual information is enlarged from about $1$ bits/symbol to about $3$ bits/symbol, while that of the CRLB for angle estimation is enlarged from  $5.7$ dB to $14.1$ dB.

\section{Conclusion}\label{section6}
In this paper, we proposed an IRS-aided NO-ISAC system, characterized its performance with a unified metric, and designed a beamforming algorithm for improving its joint performance of communication and localization. In particular, we proposed the modified CRLB metric  to characterize the joint communication and localization performance of the IRS-aided NO-ISAC system, and derived its closed-form expression. Based on the modified CRLB, we proposed a joint active and passive beamforming algorithm for balancing communication performance and localization performance. Numerical results showed that, despite the adverse effect of  signal randomness on localization, with enough time slots per time block, the IRS-aided NO-ISAC system with random communication signals can achieve comparable localization performance to the IRS-aided localization system with dedicated positioning reference signals. Moreover, it was demonstrated that both the communication and localization performance of the IRS-aided NO-ISAC system is better than those of the IRS-aided TD-ISAC system. Investigation on the trade-off between communication performance and localization performance verified the effectiveness of the proposed beamforming algorithm, revealed the existence of the trade-off region and the communication/localization saturation region, and demonstrated the benefit of applying more IRS elements  in enlarging the trade-off region.

\begin{appendices}
\section{Proof of Theorem 1}
To derive the FIM $\mathbf{J}=\mathbf{J}_P+\mathbf{J}_D$, we calculate $\mathbf{J}_P$ and $\mathbf{J}_D$ respectively in the following.
\subsection{Calculate ${\bf J}_P$}
\par  First, we calculate the $(i,j)$-th element of ${\bf J}_P$, which is given by
\begin{align}
[\mathbf{J}_P]_{i,j}= -\mathbb{E} \left( \frac{\partial ^2\ln \!\:p\left( \mathbf{x} \right)}{\partial \theta _i\partial \theta _j} \right),i,j=1,\cdots ,T+2,
\end{align}
where the probability density function $p\left( \mathbf{x} \right) $ is expressed as
\begin{align}
p\left( \mathbf{x} \right) =\left( \pi \sigma _{x}^{2} \right) ^{-T}\exp \left[ -\frac{\sum_{t=1}^T{\left| x\left( t \right) -\mu _x \right|^2}}{\sigma _{x}^{2}} \right] .\label{CRLB1}
\end{align}
Since the AoAs $\gamma _{\mathrm{I}2\mathrm{U}}^{\mathrm{D}}$ and $\varphi _{\mathrm{I}2\mathrm{U}}^{\mathrm{D}}$ are constant and the transmit signals $x\left( t \right) , t=1,\cdots ,T$ are independent of each other, only the second partial derivative of $p\left( \mathbf{x} \right)$ with respect to $x(t)$ may not be $0$, which can be calculated according to \cite{hunger2007introduction} as
\begin{align}
&\frac{\partial \ln \!\:p\left( \mathbf{x} \right)}{\partial x\left( t \right)}=-\frac{1}{\sigma _{x}^{2}}\frac{\partial \left( x\left( t \right) -\mu _x \right) \left( x\left( t \right) -\mu _x \right) ^*}{\partial x\left( t \right)}=-\frac{1}{\sigma _{x}^{2}}\left( x\left( t \right) -\mu _x \right) ^*,
\\
&\frac{\partial ^2\ln \!\:p\left( \mathbf{x} \right)}{\partial x^2\left( t \right)}=-\frac{1}{\sigma _{x}^{2}}\frac{\partial \left( x\left( t \right) -\mu _x \right) ^*}{\partial x\left( t \right)}=0.
\end{align}
Then, we obtain $\mathbf{J}_P$ as 
\begin{align}\label{JPcal}
\mathbf{J}_P=\mathbf{0}_{\left( T+2 \right) \times \left( T+2 \right)}.
\end{align}
\subsection{Calculate ${\bf J}_D$}
\par The $(i,j)$-th element of ${\bf J}_D$ can be expressed as \cite{kay1993fundamentals}
\begin{align}\label{JD}
\left[ \mathbf{J}_D \right] _{i,j}&=2\mathcal{R} \left\{ \,\frac{\partial \mathbf{h}_{x}^{\mathrm{H}}}{\partial \theta _i}\mathbf{\Sigma }^{-1}\frac{\partial \mathbf{h}_x}{\partial \theta _j} \right\} +\mathrm{tr}\left\{ \mathbf{\Sigma }^{-1}\frac{\partial \mathbf{\Sigma }}{\partial \theta _i}\mathbf{\Sigma }^{-1}\frac{\partial \mathbf{\Sigma }}{\partial \theta _j} \right\} 
\\
&=\sum_{t=1}^T{\frac{2}{\sigma _{z}^{2}}\mathcal{R} \left\{ \left[ \frac{\partial \mathbf{h}_x\left( t \right)}{\partial \theta _i} \right] ^{\mathrm{H}}\frac{\partial \mathbf{h}_x\left( t \right)}{\partial \theta _j} \right\}}\notag
\\
&=\sum_{t=1}^T{\left[ \mathbf{J}_D\left( t \right) \right] _{i,j}},i,j=1,\cdots ,T+2,\notag
\end{align}
where
\begin{align}
&\mathbf{\Sigma }=\sigma _{z}^{2}\mathbf{I}_{MT\times MT},
\\
&\left[ \mathbf{J}_D\left( t \right) \right] _{i,j}=\frac{2}{\sigma _{z}^{2}}\mathcal{R} \left\{ \left[ \frac{\partial \mathbf{h}_x\left( t \right)}{\partial \theta _i} \right] ^{\mathrm{H}}\frac{\partial \mathbf{h}_x\left( t \right)}{\partial \theta _j} \right\}.\label{JDT} 
\end{align}
\par First, we calculate the partial derivative of  $\mathbf{h}_x\left( t \right) $ with respect to $\theta _i, i=1,\cdots ,T$ (i.e., $x\left( 1 \right) ,\cdots ,x\left( T \right) $). Since the transmitted signals are independent of each other, we have
\begin{align}
\frac{\partial \mathbf{h}_x\left( t_1 \right)}{\partial x\left( t_2 \right)}=\frac{\partial \mathbf{H}_{\mathrm{I}2\mathrm{U}}\mathbf{\Theta H}_{\mathrm{B}2\mathrm{I}}\mathbf{w}x\left( t_1 \right)}{\partial x\left( t_2 \right)}=\begin{cases}
	\boldsymbol{\beta }_x,t_1=t_2\\
	0,t_1\ne t_2\\
\end{cases},\label{beta3}
\end{align}
where $\boldsymbol{\beta }_x\triangleq \mathbf{H}_{\mathrm{I}2\mathrm{U}}\mathbf{\Theta H}_{\mathrm{B}2\mathrm{I}}\mathbf{w}$.
\par Then, we calculate the partial derivative of $\mathbf{h}_x\left( t \right) =\mathbf{H}_{\mathrm{I}2\mathrm{U}}\mathbf{\Theta H}_{\mathrm{B}2\mathrm{I}}\mathbf{w}x\left( t \right) $ with respect to $\theta _{T+1}$ and $\theta _{T+2}$ (i.e., $\frac{\partial \mathbf{h}_x\left( t \right)}{\partial \gamma _{\mathrm{I}2\mathrm{U}}^{\mathrm{A}}}$ and $\frac{\partial \mathbf{h}_x\left( t \right)}{\partial \varphi _{\mathrm{I}2\mathrm{U}}^{\mathrm{A}}}$). Noticing that $d_{\mathrm{IRS}}=d_{\mathrm{user}}=\frac{\lambda}{2}$, and the URAs of both the IRS and the user lie on the $y$-$o$-$z$ plane, we have $\gamma _{\mathrm{I}2\mathrm{U}}^{\mathrm{A}}=-\gamma _{\mathrm{I}2\mathrm{U}}^{\mathrm{D}}$ and $\left| \varphi _{\mathrm{I}2\mathrm{U}}^{\mathrm{A}}-\varphi _{\mathrm{I}2\mathrm{U}}^{\mathrm{D}} \right|=\pi 
$. As such, the $(m,l)$-th element of  $\mathbf{H}_{\mathrm{I}2\mathrm{U}}$ can be compactly written as
\begin{align}
\left[ \mathbf{H}_{\mathrm{I}2\mathrm{U}} \right] _{m,l}=\alpha _{\mathrm{I}2\mathrm{U}}e^{j\left[ \left( m_y\left( m \right) -l_y\left( l \right) \right) u_{\mathrm{I}2\mathrm{U}}^{\mathrm{A}}+\left( m_z\left( m \right) -l_z\left( l \right) \right) v_{\mathrm{I}2\mathrm{U}}^{\mathrm{A}} \right]},m=1,\cdots ,M,l=1,\cdots ,L,
\end{align}
where
\begin{align}
&u_{\mathrm{I}2\mathrm{U}}^{\mathrm{A}}=\pi \cos \left( \gamma _{\mathrm{I}2\mathrm{U}}^{\mathrm{A}} \right) \sin \left( \varphi _{\mathrm{I}2\mathrm{U}}^{\mathrm{A}} \right) ,
\\
&v_{\mathrm{I}2\mathrm{U}}^{\mathrm{A}}=\pi \sin \left( \gamma _{\mathrm{I}2\mathrm{U}}^{\mathrm{A}} \right).
\end{align}
And we can obtain the partial derivative of $\mathbf{h}_x\left( t \right) $ with respect to $\gamma _{\mathrm{I}2\mathrm{U}}^{\mathrm{A}}$ and $\varphi _{\mathrm{I}2\mathrm{U}}^{\mathrm{A}}$ as
\begin{align}
&\boldsymbol{\beta }_{\gamma ,t}\triangleq \frac{\partial \mathbf{h}_x\left( t \right)}{\partial \gamma _{\mathrm{I}2\mathrm{U}}^{\mathrm{A}}}=\mathbf{H}_{\mathrm{I}2\mathrm{U},\gamma}\mathbf{\Theta H}_{\mathrm{B}2\mathrm{I}}\mathbf{w}x\left( t \right) \label{beta1},
\\
&\boldsymbol{\beta }_{\varphi ,t}\triangleq \frac{\partial \mathbf{h}_x\left( t \right)}{\partial \varphi _{\mathrm{I}2\mathrm{U}}^{\mathrm{A}}}=\mathbf{H}_{\mathrm{I}2\mathrm{U},\varphi}\mathbf{\Theta H}_{\mathrm{B}2\mathrm{I}}\mathbf{w}x\left( t \right) \label{beta2},
\end{align}
where the $(m,l)$-th elements of $\mathbf{H}_{\mathrm{I}2\mathrm{U},\gamma}\textbf{}$ and $\mathbf{H}_{\mathrm{I}2\mathrm{U},\varphi} $ are respectively defined as
\begin{align}
&\left[ \mathbf{H}_{\mathrm{I}2\mathrm{U},\gamma} \right] _{m,l}\triangleq \frac{\partial \left[ \mathbf{H}_{\mathrm{I}2\mathrm{U}} \right] _{m,l}}{\partial \gamma _{\mathrm{I}2\mathrm{U}}^{\mathrm{A}}}=j\left[ \left( m_y\left( m \right) -l_y\left( l \right) \right) \frac{\partial u_{\mathrm{I}2\mathrm{U}}^{\mathrm{A}}}{\partial \gamma _{\mathrm{I}2\mathrm{U}}^{\mathrm{A}}}+\left( m_z\left( m \right) -l_z\left( l \right) \right) \frac{\partial v_{\mathrm{I}2\mathrm{U}}^{\mathrm{A}}}{\partial \gamma _{\mathrm{I}2\mathrm{U}}^{\mathrm{A}}} \right] \left[ \mathbf{H}_{\mathrm{I}2\mathrm{U}} \right] _{m,l}\notag
\\
&\quad=j\pi \left( \left( m_z\left( m \right) -l_z\left( l \right) \right) \cos \left( \gamma _{\mathrm{I}2\mathrm{U}}^{\mathrm{A}} \right) -\left( m_y\left( m \right) -l_y\left( l \right) \right) \sin \left( \gamma _{\mathrm{I}2\mathrm{U}}^{\mathrm{A}} \right) \sin \left( \varphi _{\mathrm{I}2\mathrm{U}}^{\mathrm{A}} \right) \right) \left[ \mathbf{H}_{\mathrm{I}2\mathrm{U}} \right] _{m,l},
\\
&\left[ \mathbf{H}_{\mathrm{I}2\mathrm{U},\varphi} \right] _{m,l}\triangleq \frac{\partial \left[ \mathbf{H}_{\mathrm{I}2\mathrm{U}} \right] _{m,l}}{\partial \varphi _{\mathrm{I}2\mathrm{U}}^{\mathrm{A}}}=j\left[ \left( m_y\left( m \right) -l_y\left( l \right) \right) \frac{\partial u_{\mathrm{I}2\mathrm{U}}^{\mathrm{A}}}{\partial \varphi _{\mathrm{I}2\mathrm{U}}^{\mathrm{A}}}+\left( m_z\left( m \right) -l_z\left( l \right) \right) \frac{\partial v_{\mathrm{I}2\mathrm{U}}^{\mathrm{A}}}{\partial \varphi _{\mathrm{I}2\mathrm{U}}^{\mathrm{A}}} \right] \left[ \mathbf{H}_{\mathrm{I}2\mathrm{U}} \right] _{m,l}\notag
\\
&\quad=j\pi \left( m_y\left( m \right) -l_y\left( l \right) \right) \cos \left( \gamma _{\mathrm{I}2\mathrm{U}}^{\mathrm{A}} \right) \cos \left( \varphi _{\mathrm{I}2\mathrm{U}}^{\mathrm{A}} \right) \left[ \mathbf{H}_{\mathrm{I}2\mathrm{U}} \right] _{m,l}.
\end{align}
\par Noticing that $\mathbf{J}_D\left( t \right)$ is a real symmetric matrix, we express it as
\begin{align}\label{JDt}
\mathbf{J}_D\left( t \right) =\left[ \begin{matrix}
	\mathbf{J}_{D1}\left( t \right)&		\mathbf{J}_{D2}\left( t \right)\\
	\left[ \mathbf{J}_{D2}\left( t \right) \right] ^{\mathrm{T}}&		\mathbf{J}_{D3}\left( t \right)\\
\end{matrix} \right] ,
\end{align}
where $\mathbf{J}_{D1}\left( t \right) \in \mathbb{R} ^{T\times T}$, 
$\mathbf{J}_{D2}\left( t \right) \in \mathbb{R} ^{T\times 2}$, and $\mathbf{J}_{D3}\left( t \right) \in \mathbb{R} ^{2\times 2}$. By substituting (\ref{beta1}), (\ref{beta2}), (\ref{beta3}) into (\ref{JDT}), we obtain
\begin{align}
&\left[ \mathbf{J}_{D1}\left( t \right) \right] _{i,j}=\left\{ \begin{array}{c}
	\begin{array}{l}
	\frac{2}{\sigma _{z}^{2}}\left\| \boldsymbol{\beta }_x \right\| ^2,i=j=t\\
	0,else\\
\end{array}\\
\end{array} \right. ,
\\
&\left[ \mathbf{J}_{D2}\left( t \right) \right] _{i,j}=\begin{cases}
	\frac{2}{\sigma _{z}^{2}}\mathcal{R} \left\{ \boldsymbol{\beta }_{x}^{\mathrm{H}}\boldsymbol{\beta }_{\gamma ,t} \right\} ,i=t,j=1\\
	\frac{2}{\sigma _{z}^{2}}\mathcal{R} \left\{ \boldsymbol{\beta }_{x}^{\mathrm{H}}\boldsymbol{\beta }_{\varphi ,t}\right\} ,i=t,j=2\\
	0,else\\
\end{cases},
\\
&\mathbf{J}_{D3}\left( t \right) =\frac{2}{\sigma _{z}^{2}}\left[ \begin{matrix}
	\left\| \boldsymbol{\beta }_{\gamma ,t} \right\| ^2&		\mathcal{R} \left\{ \boldsymbol{\beta }_{\gamma ,t}^{\mathrm{H}}\boldsymbol{\beta }_{\varphi ,t} \right\}\\
	\mathcal{R} \left\{ \boldsymbol{\beta }_{\varphi ,t}^{\mathrm{H}}\boldsymbol{\beta }_{\gamma ,t} \right\}&		\left\| \boldsymbol{\beta }_{\gamma ,t} \right\| ^2\\
\end{matrix} \right] .
\end{align}
\par Substituting (\ref{JDt}) to (\ref{JD}) yields
\begin{align}\label{JDcal}
\mathbf{J}_D=\sum_{t=1}^T{\mathbf{J}_D\left( t \right)}=\left[ \begin{matrix}
	\mathbf{J}_{D1}&		\mathbf{J}_{D2}\\
	\mathbf{J}_{D2}^{\mathrm{T}}&		\mathbf{J}_{D3}\\
\end{matrix} \right] ,
\end{align}
where
\begin{align}
&\mathbf{J}_{D1}=\frac{2}{\sigma _{z}^{2}}\left\| \boldsymbol{\beta }_x \right\| ^2\mathbf{I}_{T\times T},
\\
&\mathbf{J}_{D2}=\frac{2}{\sigma _{z}^{2}}\left[ \begin{matrix}
	\mathcal{R} \left\{ \boldsymbol{\beta }_{x}^{\mathrm{H}}\boldsymbol{\beta }_{\gamma ,1} \right\}&		\mathcal{R} \left\{ \boldsymbol{\beta }_{x}^{\mathrm{H}}\boldsymbol{\beta }_{\varphi ,1} \right\}\\
	\vdots&		\vdots\\
	\mathcal{R} \left\{ \boldsymbol{\beta }_{x}^{\mathrm{H}}\boldsymbol{\beta }_{\gamma ,T} \right\}&		\mathcal{R} \left\{ \boldsymbol{\beta }_{x}^{\mathrm{H}}\boldsymbol{\beta }_{\varphi ,T} \right\}\\
\end{matrix} \right] ,
\\
&\mathbf{J}_{D3}=\frac{2}{\sigma _{z}^{2}}\sum_{t=1}^T{\left[ \begin{matrix}
	\!\left\| \boldsymbol{\beta }_{\gamma ,t} \right\| ^2&		\!\mathcal{R} \left\{ \boldsymbol{\beta }_{\gamma ,t}^{\mathrm{H}}\boldsymbol{\beta }_{\varphi ,t}\! \right\}\\
	\mathcal{R} \left\{ \boldsymbol{\beta }_{\varphi ,t}^{\mathrm{H}}\boldsymbol{\beta }_{\gamma ,t}\! \right\}&		\left\| \boldsymbol{\beta }_{\varphi ,t} \right\| ^2\\
\end{matrix} \right]}.\label{JD3}
\end{align}
\par Finally, combing (\ref{JPcal}) and (\ref{JDcal}), we obtain the FIM with respect to $\boldsymbol{\theta }$ as
\begin{align}
    \mathbf{J}=\mathbf{J}_P+\mathbf{J}_D=\mathbf{J}_D.
\end{align}

\end{appendices}
\bibliographystyle{IEEEtran}
\bibliography{references}{}

\end{document}